\documentclass[12pt]{scrartcl}
\usepackage{a4}
\usepackage{amsmath}
\usepackage{amssymb}
\usepackage{amsfonts}
\usepackage{mathrsfs}
\usepackage{dsfont}
\usepackage{latexsym}
\usepackage{xcolor}
\usepackage{bbm,exscale}
\definecolor{Myblue}{rgb}{0,0,0.6}
\usepackage[colorlinks,citecolor=Myblue,linkcolor=Myblue,urlcolor=Myblue,pdfpagemode=UseNone]{hyperref}
\usepackage{amsthm}
\usepackage[square,numbers,sort&compress]{natbib}
\usepackage[all,cmtip]{xy}
\usepackage{tikz}
\usetikzlibrary{decorations.pathmorphing}
\usetikzlibrary{calc}
\usetikzlibrary{decorations.markings}
\usetikzlibrary{fadings,decorations.pathreplacing}
\usetikzlibrary{matrix,arrows}
\usepackage{bbding}
\usetikzlibrary{arrows,calc,decorations.pathreplacing,decorations.markings,shapes.geometric,shadows}

\tikzset{
    string/.style={draw=#1, postaction={decorate}, decoration={markings,mark=at position .51 with {\arrow[draw=#1]{>}}}},
    costring/.style={draw=#1, postaction={decorate}, decoration={markings,mark=at position .51 with {\arrow[draw=#1]{<}}}},
    ostring/.style={draw=#1, postaction={decorate}, decoration={markings,mark=at position .47 with {\arrow[draw=#1]{>}}}},
    ustring/.style={draw=#1, postaction={decorate}, decoration={markings,mark=at position .56 with {\arrow[draw=#1]{>}}}},
    oostring/.style={draw=#1, postaction={decorate}, decoration={markings,mark=at position .43 with {\arrow[draw=#1]{>}}}},
    uustring/.style={draw=#1, postaction={decorate}, decoration={markings,mark=at position .59 with {\arrow[draw=#1]{>}}}},
    directed/.style={string=blue!50!black}, 
    odirected/.style={ostring=blue!50!black}, 
    udirected/.style={ustring=blue!50!black}, 
    oodirected/.style={oostring=blue!50!black}, 
    uudirected/.style={uustring=blue!50!black},     
    redirected/.style={costring= blue!50!black},
    redirectedgreen/.style={costring= green!50!black},
    directedgreen/.style={string= green!50!black},
    redirectedlightgreen/.style={costring= green!65!black},
    directedlightgreen/.style={string= green!65!black},
}

\tikzset{-dot-/.style={decoration={
  markings,
  mark=at position 0.5 with {\fill circle (1.875pt);}},postaction={decorate}}}

\tikzset{
	Fdot/.style={circle, draw, fill, inner sep=0pt}, 
	Odot/.style={circle, draw, inner sep=0.1pt, minimum size=0.1cm}
	}

\def\nicepalecolourscheme{\shadedraw[top color=blue!12, bottom color=blue!12, draw=white]}

  \vfuzz \hfuzz
\makeatletter
\newcommand{\raisemath}[1]{\mathpalette{\raisem@th{#1}}}
\newcommand{\raisem@th}[3]{\raisebox{#1}{$#2#3$}}
\makeatother

\renewcommand{\H}{\mathcal H}

\newcommand{\E}{\text{e}}
\newcommand{\I}{\text{i}}
\newcommand{\B}{\mathcal{B}}

\newcommand{\Beq}{\B_{\mathrm{eq}}}
\newcommand{\C}{\mathds{C}}

\newcommand{\Z}{\mathds{Z}}

\def\1{\ifmmode\mathrm{1\!l}\else\mbox{\(\mathrm{1\!l}\)}\fi}

\newcommand{\be}{\begin{equation}}
\newcommand{\ee}{\end{equation}}
\newcommand{\bes}{\begin{equation*}}
\newcommand{\ees}{\end{equation*}}

\newcommand{\Hom}{\operatorname{Hom}}
\newcommand{\End}{\operatorname{End}}
\newcommand{\modu}{\operatorname{mod}}

\newcommand{\hmf}{\operatorname{hmf}}

\newcommand{\ev}{\operatorname{ev}}

\newcommand{\coev}{\operatorname{coev}}

\def\lra{\longrightarrow}
\def\lmt{\longmapsto}
\DeclareMathOperator{\tr}{tr}
\DeclareMathOperator{\str}{str}

\newcommand{\pirr}{\pi^{\text{RR}}_A}

\DeclareMathOperator*{\eq}{=}

\newcommand{\AGC}{A_G^c}

\newcommand{\dX}{{}^\dagger\hspace{-1.8pt}X}
\newcommand{\dA}{{}^\dagger\hspace{-1.8pt}A}
\newcommand{\dsX}{{}^\dagger\hspace{-1.8pt}\mathcal{X}}
\newcommand{\deqX}{{}^\star\hspace{-1.8pt}X} 
\newcommand{\dseqX}{{}^\star\hspace{-1.8pt}\mathcal{X}}

\newcommand\arxiv[2]      {\href{http://arXiv.org/abs/#1}{#2}}
\newcommand\doi[2]        {\href{http://dx.doi.org/#1}{#2}}
\newcommand\httpurl[2]    {\href{http://#1}{#2}}

\allowdisplaybreaks

\deffootnote[1em]{1em}{1em}
{\textsuperscript{\thefootnotemark}}

\theoremstyle{definition}
\newtheorem{definition}{Definition}
\newtheorem{proposition}[definition]{Proposition}
\newtheorem{theorem}[definition]{Theorem}
\newtheorem{lemma}[definition]{Lemma}

\newtheorem{remark}[definition]{Remark}

\newtheorem{example}[definition]{Example}

\numberwithin{equation}{section}
\numberwithin{definition}{section}
\numberwithin{figure}{section}

\newcommand\void[1]{}

\begin{document}

\title{Discrete torsion defects}

\author{Ilka Brunner$^*$ \quad Nils Carqueville$^\dagger$ \quad Daniel Plencner$^*$
\\[0.5cm]
 \normalsize{\tt \href{mailto:ilka.brunner@physik.uni-muenchen.de}{ilka.brunner@physik.uni-muenchen.de}} \\
  \normalsize{\tt \href{mailto:nils.carqueville@univie.ac.at}{nils.carqueville@univie.ac.at}} \\
  \normalsize{\tt \href{mailto:daniel.plencner@physik.uni-muenchen.de}{daniel.plencner@physik.uni-muenchen.de}}\\[0.1cm]
  {\normalsize\slshape $^*$Arnold Sommerfeld Center for Theoretical Physics, LMU M\"unchen}\\[-0.1cm]
  {\normalsize\slshape $^*$Excellence Cluster Universe, Garching}\\[-0.0cm]
  {\normalsize\slshape $^\dagger$Erwin Schr\"odinger Institute \& Fakult\"at f\"ur Mathematik, University of Vienna}\\[-0.1cm]
}
\date{}
\maketitle

\vspace{-11.8cm}
\hfill {\scriptsize LMU-ASC 22/14}

\vspace{11cm}

\begin{abstract}
Orbifolding two-dimensional quantum field theories by a symmetry group can involve a choice of discrete torsion. 
We apply the general formalism of `orbifolding defects' to study and elucidate discrete torsion for topological field theories. 
In the case of Landau-Ginzburg models only the bulk sector had been studied previously, and we re-derive all known results. 
We also introduce the notion of `projective matrix factorisations', show how they naturally describe boundary and defect sectors, 
and we further illustrate the efficiency of the defect-based approach by explicitly computing RR charges. 

Roughly half of our results are not restricted to Landau-Ginzburg models but hold more generally, for any topological field theory. 
In particular we prove that for a pivotal bicategory, any two objects of its orbifold completion that have the same base are orbifold equivalent. Equivalently, from any orbifold theory (including those based on nonabelian groups) the original unorbifolded theory can be obtained by orbifolding via the `quantum symmetry defect'. 
\end{abstract}

\thispagestyle{empty}
\newpage

\tableofcontents

\section{Introduction}\label{sec:introduction}

Orbifold constructions allow to obtain new theories from old ones, provided the initial theory is equipped with a symmetry group $G$. Orbifolds have been studied extensively in particular for two-dimensional conformal and topological field theories. On surfaces without boundary, the construction is performed in two steps: First, new twisted sectors are added to the theory. In string theory, one understands them as strings that close only up to an action of an element of the group.  As a second step one projects to $G$-invariant states. 

While the addition of twisted sectors is a universal procedure, the projection involves in general some choice. To specify  a projection, one needs an action of~$g$ in the $h$-twisted sector, for all $g,h\in G$. This is generally not uniquely defined, instead, given a consistent action of~$g$ in the $h$-twisted sector ${\cal H}_h$, one can ask under which circumstances the action
$$
\hat{g} : {\cal H}_h \lra {\cal H}_h 
\, , \quad 
\hat{g} = \varepsilon (g,h) \cdot g \, ,
$$
with $\varepsilon(g,h)$ a phase, leads to an alternative consistent and inequivalent projection. There are constraints originating from the consistency of the theory on higher genus Riemann surfaces \cite{VafaDiscreteTorsion}.
The choice of phase~$\varepsilon$ is usually referred to as `discrete torsion' and is classified by $H^2(G, U(1))$, the second cohomology of the group with values in $U(1)$. 
Elements $c\in H^2(G, U(1))$ can be viewed as maps from $G\times G$ to $U(1)$ which are in the kernel of an operator~$\delta$ with
$$
(\delta c)(g,h,k) = \frac{ c(g, hk) \, c(h, k)}{c(g, h) \, c(gh, k)} \, . 
$$
This gives the condition 
$
c(g,hk) c(h,k) = c(g,h) c(gh,k) 
$ 
for all $g,h,k\in G$, and furthermore two cocycles $c,c'$ are identified if they differ by a coboundary in the sense that there is a map $\alpha: G \rightarrow U(1)$ such that $c(g,h) = \tfrac{\alpha(g) \alpha(h)}{\alpha(gh)} c'(g,h)$. 
For abelian groups, it is well-known that the choice of orbifold projection is related to the cocycle as
$$
\varepsilon(g, h) = \frac{c(g,h)}{c(h,g)} \, . 
$$

\medskip

To include worldsheets with boundary in orbifold backgrounds, one 
starts with consistent boundary conditions in the parent theory and, if necessary, renders them
invariant by summing over the images under the group action. This  summation is already fully determined by the orbifold action on bulk fields. On the invariant boundary conditions, one then implements an action of the orbifold group on open string states. This action is not fully determined by the bulk projection, but involves a choice of a representation on the Chan-Paton labels. 

It was argued in \cite{Dou,DF,gomis,sharpe,Aspinwall:2000xv,hauer-krogh} that for theories with discrete torsion the representations on Chan-Paton labels are projective.\footnote{However, examples of D-branes exist where the representation is more complicated \cite{gaberdiel, cg0101143}.}
In fact, the space $H^2(G, U(1))$ also classifies central extensions 
$$
1 \lra U(1) \lra \widehat{G} \lra {G} \lra 1 
$$
of the group~$G$. 
A projective representation is then given by a map $ \gamma: G\to \widehat{G}$ that inverts the respective map in the short exact sequence. One finds that
$$
\gamma (g) \gamma (h) = c(g,h) \cdot \gamma (gh)
$$
for an appropriate $c\in H^2(G, U(1))$. 

\medskip

In the present paper, we apply the general theory \cite{ffrs0909.5013,cr1210.6363, BCP1} of orbifolds via defects to the special case of discrete torsion orbifolds. 
For every $g\in G$ there is a defect ${}_g I$ which implements the action of~$g$. For bulk fields this means that one obtains the image of the bulk field under~$g$ by pulling the defect across the insertion point of the bulk field. Fields in the~$g$-twisted sector are regarded as defect changing operators between ${}_g I$ and the identity defect $I={}_e I$. 

To formulate the orbifold theory, one considers the superposition 
$$
A_G= \bigoplus_{g\in G} {}_g I \, . 
$$
Correlation functions of the orbifold theory are then  obtained as correlators of the original theory with a sufficiently fine network of $A_G$ defect lines \cite{ffrs0909.5013, dkr1107.0495}. Consistency requires that the correlator is independent of the chosen defect network; 
this invariance in particular encodes an averaging over twisted sectors, see e.\,g.~\cite[Rem.\,3.6]{cr1210.6363}. 
It was shown in \cite{ffrs0909.5013} for rational conformal field theories and in \cite{cr1210.6363} for topological field theories that this translates into invariance under certain local changes of the defect network that are encoded in the structure of a (symmetric) separable Frobenius algebra $A_G$. In particular, the product of this algebra consists of a defect junction with two ingoing and one outgoing defect line, at which a suitable defect changing field must be inserted. For rational CFTs one finds \cite{tft3, ffrs0909.5013} that the different consistent choices again correspond to $H^2(G,U(1))$, such that theories with discrete 
torsion are in this framework described by a choice of defect changing field for trivalent junctions. To distinguish different choices, we write $A_G^c$ for the defect $A_G$ when viewed as the Frobenius algebra associated to $c \in H^2(G,U(1))$. 

Note that this framework applies quite generally, in particular both abelian and nonabelian orbifolds can be treated on the same footing. 
In fact, one does not necessarily have to restrict to orbifolding defects of the form $A_G$ originating from a symmetry group of the parent theory. The only requirement is that the orbifolding defect~$A$ has the structure of a separable Frobenius algebra. 
In general the algebra~$A$ replaces the group~$G$, and it is the (only) fundamental object to build the orbifold theory; everything else can be derived from it. 
The special cases $A=A_G^c$ are precisely how discrete torsion orbifolds fall into the general framework \cite{ffrs0909.5013}. 

Given any orbifolding defect~$A$, to include consistent boundary conditions one has to require that $A$-lines can end on boundaries, where the junction is again interpreted as the insertion point of a field. It was shown in \cite{tft1} for rational conformal field theories and in \cite{cr1210.6363} for topological field theories that a consistent class of orbifold boundary conditions is given by $A$-modules in the parent theory. The module structure is in particular important to show that the bulk-boundary maps of twisted sector fields are well-defined in the orbifold theory. The general formalism developed in \cite{tft1, cr1210.6363, BCP1} guarantees that the correlators obtained from the parent theory with defects satisfy all axioms of open/closed topological field theory. 

\medskip

In this paper, we will apply the general defect formalism to discuss orbifolds of  Landau-Ginzburg models with discrete torsion, i.\,e.~the case $A=A_G^c$. For surfaces without boundaries, this was first discussed in \cite{IntriligatorVafa1990}, where in particular the different projections on the closed string sector were obtained. As a first step we will rederive these results from the defect perspective.

Boundary conditions in Landau-Ginzburg models are given by matrix factorisations \cite{kl0210, bhls0305, l0312}. In orbifold theories without discrete torsion, these matrix factorisations have to be equivariant \cite{add0401}, but the case with discrete torsion has not previously been discussed in the literature. In this paper we follow the general defect-based formalism which tells us that boundary conditions now are $A_G^c$-modules. We shall show that these are the same as what we call `projective matrix factorisations', i.\,e.~equivariant factorisations furnished with a projective representation of the orbifold group. 

Defects in discrete torsion orbifolds can be comfortably discussed in our setting too. They are described as $A_G^c$-bimodules. An interesting structure on the level of defects is the fusion product, which arises when two defects merge to form a new defect. In topological theories, this product is non-singular and in the case of (unorbifolded) Landau-Ginzburg models it is given by the tensor product of two matrix factorisations. 
In an $A$-orbifold this tensor product has to be taken over the algebra~$A$. 
For discrete torsion orbifolds one simply chooses $A=A_G^c$, and we will see that fusion products can be computed very explicitly. Furthermore, as a corollary of the general discussion of fusion in \cite{cr1210.6363, BCP1} one obtains the Cardy condition for discrete torsion orbifolds: 
two-point correlators of boundary-bulk maps are given by traces in the open sector. 
For non-projective $G$-orbifolds this leads to the Landau-Ginzburg version of the equivariant Hirzebruch-Riemann-Roch theorem first proved in \cite{pv1002.2116}; we generalise the Cardy condition to the case of projective matrix factorisations (but we do not further evaluate either side of the equality). 

Summarising, in the defect approach to orbifolds, all the necessary information is encoded in the orbifolding defect $A_G^c$. 
In the next section, we lay out the general construction for Landau-Ginzburg models with discrete torsion. 
In Section~\ref{sec:disccorrelators}, we discuss the consequences of $G$-invariance for disc correlators, and we show how the computation of RR charges simplifies for D-branes carrying induced representations. 
Finally in Section~\ref{sec:orbifoldequivalences} we prove a theorem on orbifold equivalences in a general bicategorical setting. This may be of general interest, 
but our discussion at the end of Sections~\ref{subsec:defectsec} and~\ref{sec:orbifoldequivalences} also illustrates its significance for TFTs such as Landau-Ginzburg models: orbifolding the $G$-orbifold of a theory by its quantum symmetry defect reproduces the original theory, a result which previously was known only for abelian~$G$.

\section{Discrete torsion Landau-Ginzburg orbifolds}\label{sec:DTfromAGc}

\subsection{Orbifolding defect}\label{subsec:orbidefect}

We fix a Landau-Ginzburg potential $W \in \C[x_1,\ldots,x_n]$ and recall that a defect~$A$ in the theory~$W$ is a matrix factorisation of $W(x') - W(x)$ \cite{br0707.0922}. 
A particular example is the invisible or identity defect $I=I_W$ with twisted differential 
\be\label{eq:IdDef}
\bigotimes_{i=1}^n 
\begin{pmatrix}
0 & x'_i - x_i \\
\tfrac{W(x_1,\ldots,x_{i-1},x'_i,\ldots,x'_n) - W(x_1,\ldots,x_{i},x'_{i+1},\ldots,x'_n)}{x'_i - x_i} & 0
\end{pmatrix}
.
\ee
It has a natural left action $\lambda_X: I \otimes X \rightarrow X$ on any defect~$X$ with the theory~$W$ to its left. For details and background we refer to \cite{br0707.0922, cm1208.1481, cr1210.6363}. 

For~$A$ to give rise to a consistent orbifold theory it must have the structure of a separable Frobenius algebra \cite{ffrs0909.5013, cr1210.6363, BCP1}. The precise definition can e.\,g.~be found in \cite[Sect.\,2.2]{cr1210.6363}; most relevant for our purposes is that~$A$ comes with a trivalent junction field 
$$
\mu = 
\begin{tikzpicture}[very thick,scale=0.75,color=green!50!black, baseline=0.4cm]
\draw[-dot-] (3,0) .. controls +(0,1) and +(0,1) .. (2,0);
\draw (2.5,0.75) -- (2.5,1.5); 
\end{tikzpicture} 
: AÊ\otimes A \lra A \, ,
$$
the `multiplication' of the algebra~$A$. 

Every group orbifold can be described in terms of the associated defect. For a finite symmetry group~$G$ it is given by
\be\label{eq:AG}
A_G = \bigoplus_{g\in G} {}_g I 
\quad \text{with} \quad
\mu = \sum_{g,h \in G} {}_g (\lambda_{{}_h I}) \, . 
\ee
Here the notation ${}_g(-)$ means that in the respective matrix such as~\eqref{eq:IdDef} every `left' or `primed' variable~$x'_i$ is to be replaced by $g(x'_i)$. 

More generally, we can make use of the $\Z_2$-shift functor $[-]$ and consider the defects 
\be\label{eq:AGK}
A_{G,\kappa} = \bigoplus_{g\in G} {}_g I [\kappa_g]
\quad \text{with} \quad
\mu = \sum_{g,h \in G} {}_g (\lambda_{{}_h I[\kappa_h]})[\kappa_g] 
\ee
for all group homomorphisms $\kappa: G \rightarrow \Z_2$. 
The reason that there are several defects $A_{G,\kappa}$ associated to a symmetry group~$G$ is that in Landau-Ginzburg models, the action of $g\in G$ on untwisted (c,c) fields (i.\,e.~polynomials in~$x_i$) does not uniquely determine the action on all bulk fields in the orbifold theory. Rather, as was shown in~\cite{IntriligatorVafa1990}, the action on untwisted RR ground states is determined only up to a sign $(-1)^{\kappa_g}$. Below in Remark~\ref{rem:AGkappa} we will recover these signs from our defect-driven perspective by `wrapping $A_{G,\kappa}$ around RR ground states'. 

Another way to understand the form of $A_{G,\kappa}$ is via the $\Z_2$-symmetry generated by $(-1)^{F_s}$ that maps all fields $\phi$ in the RR sector to $-\phi$. Taking $\kappa \neq 0$ corresponds to orbifolding by $G'=\{(g,\kappa_g) \,|\, g\in G\} \subset G\times \Z_2$, and the orbifold defect associated to the $\Z_2$-symmetry is $A_{\langle (-1)^{F_s}\rangle}= I \oplus I[1]$. 
We then see that $A_{G,\kappa} \cong A_{G'}=\bigoplus_{g\in G} {}_g I \otimes I[\kappa_g] \subset A_G \otimes A_{\langle (-1)^{F_s} \rangle}$. All results such as Proposition~\ref{prop:AGcsFA} and Theorem~\ref{thm:modAGc} below hold for any $A_{G,\kappa}$, but to avoid clutter we only treat the case $\kappa=0$ explicitly outside of Remark~\ref{rem:AGkappa} and Example~\ref{ex:inducedcharges}. 

\medskip

It is natural to ask whether the underlying defect may be endowed with more than one separable Frobenius structure. This would mean that the same group~$G$ can give rise to different orbifold theories. Of course we expect this to be related to discrete torsion, and below we will make this relation precise. To do so we recall that $H^2(G,U(1))$, the second group cohomology with values in $U(1)$, is given by maps $c: G \times G \rightarrow U(1)$ with 
\be\label{eq:cinH2}
c(g,hk) \, c(h,k) = c(g,h) \, c(gh,k) 
\quad \text{for all } g,h,k \in G
\ee
modulo an equivalence relation via 1-coboundaries, with $c\sim c'$ if there is a map $\alpha: G \rightarrow U(1)$ such that $c(g,h) = \tfrac{\alpha(g) \alpha(h)}{\alpha(gh)} c'(g,h)$. We will write~$c$ both for the cohomology class and its representative. 

\begin{proposition}\label{prop:AGcsFA}
If the
graded
matrix factorisation~$A_G$ satisfies $\text{dim} \Hom^0(I, A_G)=1$, then its separable Frobenius structures are classified by $H^2(G,U(1))$. 
In particular, the multiplication associated to $c \in H^2(G,U(1))$ is
\be\label{eq:muGc}
\mu_G^c = \sum_{g,h \in G} c(g,h) \cdot {}_g \! \left( \lambda_{{}_h I} \right) . 
\ee
\end{proposition}

\begin{proof}
By assumption, the space of degree zero maps in $\Hom({}_g I \otimes {}_h I,{}_k I) \cong \Hom(I, {}_{h^{-1}g^{-1}k} I)$ is one-dimensional if $k=gh$, and trivial otherwise. Any map $A_G \otimes A_G \rightarrow A_G$ is then of the form $\sum_{g,h \in G} c(g,h) \cdot \mu_{g,h}$ where~$c$ is any complex function on $G\times G$ and $\mu_{g,h}: {}_g I \otimes {}_h I \rightarrow {}_{gh} I$.   In particular, we can take $\mu_{g,h} = {}_g(\lambda_{{}_h I})$. Similarly, the comultiplication must be of the form $\tfrac{1}{|G|} \sum_{g,h \in G} \widetilde c(g,h) {}_g(\lambda^{-1}_{{}_h I})$ for some complex function~$\widetilde c$. 

The remainder of the proof proceeds as that of \cite[Prop.\,7.1(i)]{cr1210.6363}, with the following additions: 
(co)unitality forces $c(g,e)=c(e,g)=\widetilde c(g,e)=\widetilde c(e,g)=1$; the Frobenius property leads in particular to $c(g,h) \widetilde c(h, k) = c(g,hk) \widetilde c (gh,k)$, hence $h=e$ yields $\widetilde c(g,k) = c^{-1}(g,k)$; separability is then automatic; (co)associativity leads to the cocycle condition~\eqref{eq:cinH2}; the independence of the choice of representative of the cohomology class is due to the fact that every Frobenius morphism is an isomorphism, and isomorphisms ${}_g I \rightarrow {}_h I$ are rescalings of the form $\delta_{g,h} \alpha(g)$, so isomorphic Frobenius structures differ by $\tfrac{\alpha(g) \alpha(h)}{\alpha(gh)}$. 
\end{proof}

The assumption of Proposition~\ref{prop:AGcsFA} is in particular satisfied if $A_G$ is a graded matrix factorisations with $0<|x_i|<1$ for all variables $x_i$~\cite[App.\,A]{br0707.0922}. We will write $A_G^c$ for the defect $A_G$ if we want to stress that its algebra structure is that coming from $c \in H^2(G,U(1))$. Note that thanks to Proposition~\ref{prop:AGcsFA} all the general results of \cite{cr1210.6363, BCP1} hold for the $A_G^c$-orbifold theory; some of these will be discussed in the next sections, along with additional results.

\subsection{Bulk sector}\label{subsec:bulksec}

As first observed in \cite{kr0405232}, the bulk space of states of an unorbifolded Landau-Ginzburg model is isomorphic to the space of endomorphisms of the identity defect. This is true much more generally. In our present case it translates into the fact that the bulk space (of (c,c) fields) in the discrete torsion orbifold theory is given by $\End_{\AGC,\AGC}(\AGC)$, i.\,e.~endomorphisms of $\AGC$ viewed as a bimodule over itself. $\End_{\AGC,\AGC}(\AGC)$ is simply made up of those maps $\phi \in \End(\AGC)$ in the unorbifolded theory which commute with $\mu_G^c$, 
$$
\begin{tikzpicture}[very thick,scale=0.75,color=green!50!black, baseline=0.4cm]
\draw[-dot-] (3,0) .. controls +(0,1) and +(0,1) .. (2,0);
\draw (2.5,0.75) -- (2.5,1.5); 
\fill (2,0) circle (2.5pt) node[left] (alpha) {{\small $\phi$}};
\draw (2,0) -- (2,-0.5); 
\draw (3,0) -- (3,-0.5); 
\end{tikzpicture} 
=
\begin{tikzpicture}[very thick,scale=0.75,color=green!50!black, baseline=0.4cm]
\draw[-dot-] (3,0) .. controls +(0,1) and +(0,1) .. (2,0);
\draw (2.5,0.75) -- (2.5,1.5); 
\fill (2.5,1.15) circle (2.5pt) node[left] (alpha) {{\small $\phi$}};
\draw (2,0) -- (2,-0.5); 
\draw (3,0) -- (3,-0.5); 
\end{tikzpicture} 
=
\begin{tikzpicture}[very thick,scale=0.75,color=green!50!black, baseline=0.4cm]
\draw[-dot-] (3,0) .. controls +(0,1) and +(0,1) .. (2,0);
\draw (2.5,0.75) -- (2.5,1.5); 
\fill (3,0) circle (2.5pt) node[right] (beta) {{\small $\phi$}};
\draw (2,0) -- (2,-0.5); 
\draw (3,0) -- (3,-0.5); 
\end{tikzpicture} 
\, . 
$$
Here and below all trivalent vertices on $\AGC$-lines are secretly labelled by the (co)multiplication of $\AGC$. 

As explained in \cite[Sect.\,3.2]{BCP1} an equivalent way to describe $\End_{\AGC,\AGC}(\AGC)$ is as the image of a certain projector $\pi_{\AGC}$ on $\Hom(I,\AGC)$, i.\,e.~on the space of junction fields sitting at one end of the defect $\AGC$. 
The projector is 
\be\label{eq:piAGc}
\pi_{\AGC} : 
\Hom(I,\AGC) = \bigoplus_{g \in G} \Hom(I,{}_g I) \ni 
\begin{tikzpicture}[very thick,scale=0.75,color=green!50!black, baseline]
\fill (0,-0.5) circle (2.5pt) node[left] (D) {{\small $\alpha$}};
\draw (0,-0.5) -- (0,0.6); 
\end{tikzpicture} 
\lmt 
\begin{tikzpicture}[very thick,scale=0.75,color=green!50!black, baseline]
\draw (0,0) -- (0,1.25);
\fill (0,0) circle (2.5pt) node[left] {{\small $\alpha$}};
\draw (0,0.8) .. controls +(-0.9,-0.3) and +(-0.9,0) .. (0,-0.8);
\draw (0,-0.8) .. controls +(0.9,0) and +(0.7,-0.1) .. (0,0.4);
\fill (0,-0.8) circle (2.5pt) node {};
\fill (0,0.4) circle (2.5pt) node {};
\fill (0,0.8) circle (2.5pt) node {};
\draw (0,-1.2) node[Odot] (unit) {};
\draw (0,-0.8) -- (unit);
\end{tikzpicture}
\in \Hom(I,\AGC)
\ee
where 
$
\begin{tikzpicture}[thick,scale=0.4,color=green!50!black, baseline=-0.2cm]
\draw (0,-0.5) node[Odot] (D) {}; 
\draw (D) -- (0,0.3); 
\end{tikzpicture} 
$
is the inclusion of~$I$ into $\AGC$. 
In this presentation of the bulk space the twisted sectors $\pi_{\AGC} (\Hom(I,{}_g I))$ are manifest. 

\medskip

We can recover the conventional description of discrete torsion orbifolds of \cite{IntriligatorVafa1990} by computing the orbifold projection of $\pi_{\AGC}$. According to \eqref{eq:piAGc} the image of a field $\alpha_h \in \Hom(I,{}_h I)$ in the~$h$-twisted sector is the sum over all $g\in G$ of 
\be\label{eq:hwrapsg}
\begin{tikzpicture}[very thick,scale=1.35,color=green!65!black, baseline]
\draw (0,-0.3) -- (0,1.25);
\fill (0,-0.3) circle (1.75pt) node[left] {{\small $\alpha_h$}};
\draw (0,0.8) .. controls +(-0.9,-0.3) and +(-0.9,0) .. (0,-0.8);
\draw (0,-0.8) .. controls +(0.9,0) and +(0.7,-0.1) .. (0,0.4);
\fill (0,-0.8) circle (1.75pt) node {};
\fill (0,0.4) circle (1.75pt) node {};
\fill (0,0.8) circle (1.75pt) node {};
\draw (0,-1.2) node[Odot] (unit) {};
\draw (0,-0.8) -- (unit);
\fill (-0.59,0.1) circle (0pt) node[left] {{\tiny $g\vphantom{g^{-1}}$}};
\fill (0.29,0.1) circle (0pt) node[left] {{\tiny $h\vphantom{g^{-1}}$}};
\fill (1.1,0.1) circle (0pt) node[left] {{\tiny $g^{-1}$}};
\fill (0.67,0.65) circle (0pt) node[left] {{\tiny $hg^{-1}$}};
\fill (0.78,1.15	) circle (0pt) node[left] {{\tiny $ghg^{-1}$}};
\fill (0.28,-0.98) circle (0pt) node[left] {{\tiny $e\vphantom{g^{-1}}$}};
\fill[color=gray] (-0.6,-0.95) circle (0pt) node {{\tiny $c(g,g^{-1})^{-1}$}};
\fill[color=gray] (0.565,0.88) circle (0pt) node {{\tiny $c(g,hg^{-1})$}};
\fill[color=gray] (0.52,0.48) circle (0pt) node {{\tiny $c(h,g^{-1})$}};
\end{tikzpicture}
\ee
where in grey we have indicated the additional factors that arise if $c\neq 1$. 
The expression~\eqref{eq:hwrapsg} is to be interpreted as `wrapping the symmetry defect~${}_g I$ around the bulk field~$\alpha_h$'. This is the same as the action of~$g$ on~$\alpha_h$, and we read off from~\eqref{eq:hwrapsg} that it differs from the case without discrete torsion by a factor 
$$
\varepsilon (g,h) 
= 
\frac{c(g,hg^{-1}) \, c(h,g^{-1})}{c(g,g^{-1})}
\, . 
$$
Using~\eqref{eq:cinH2} and $c(k,e)=1$ for all $k\in G$ we find that 
$$
\varepsilon (g,h) 
= 
\frac{c(g,h) \, c(gh,g^{-1})}{c(g,g^{-1})} = \frac{c(g,h) \, c(g^{-1},g)}{c(g,g^{-1}) \, c(ghg^{-1},g)}=\frac{c(g,h)}{c(ghg^{-1},g)} \, .
$$
We thus recover the familiar relation $\varepsilon (g,h) = \tfrac{c(g,h)}{c(h,g)}$ for abelian~$G$. 
An analogous computation shows that 
the projector on RR ground states~\cite[Sect.\,3.2]{BCP1} is modified by the same factor of $\varepsilon(g,h)$.
Here we are led to obtain it as a consequence of the general theory of \cite{cr1210.6363, BCP1} and the single choice of fundamental defect $\AGC$. 

\begin{remark}\label{rem:AGkappa}
For a complete comparison with the conventional description of~\cite{IntriligatorVafa1990}, we also need to discuss the action of the orbifolding defect $A_{G,\kappa}^c = \bigoplus_{g\in G} {}_g I [\kappa_g]$. First, we note that $A_{G,\kappa}^c$ carries the structure of a separable Frobenius algebra with (co)multiplication defined by
\be
\mu = \sum_{g,h \in G} c(g,h) \cdot {}_g (\lambda_{{}_h I[\kappa_h]})[\kappa_g] 
\, , \quad 
\Delta = \frac{1}{|G|} \sum_{g,h \in G} c(g,h)^{-1} \cdot {}_g (\lambda_{{}_h I[\kappa_h]})^{-1}[\kappa_g] 
\ee
and (co)unit the same as for $A_G$. 
The Nakayama automorphism (see~\eqref{eq:Nakayama}) is 
\be\label{eq:NakayamaAGkappa}
\gamma_{A_{G,\kappa}^c}=\sum_{g\in G} (-1)^{\kappa_g} \det(g) \cdot 1_{{}_g I [\kappa_g] } \, . 
\ee
The necessary checks proceed along the lines of \cite[Prop.\,7.1(i)]{cr1210.6363} and \cite[Ex.\,3.1]{BCP1}; in particular associativity of~$\mu$ forces~$\kappa: G\rightarrow \Z_2$ to be a group homomorphism.

Unprojected twisted fields in the $A_{G,\kappa}^c$-orbifold theory are described by the space $\Hom(I,A_{G,\kappa}^c)$. Compared with $\Hom(I,A_G^c)$, which as a vector space is isomorphic to $\Hom(I,A_{G,\kappa}^c)$, we see that $\kappa \neq 0$ has the effect of shifting the $\Z_2$-degree of a field in the $g$-twisted sector by $\kappa_g$. Similarly, the orbifold projectors on (c,c) fields and RR ground states are decorated with additional signs originating from the $\kappa$-shifts. Computing the action of $g\in G$ on an $h$-twisted field along the lines of~\cite[App.\,A.2]{BCP1}, we obtain the sign $(-1)^{\kappa_g \kappa_h}$ for the (c,c)-projector and $(-1)^{\kappa_g(1+\kappa_h)}$ for the RR-projector (using~\eqref{eq:NakayamaAGkappa} for the latter). 
The action of $g\in G$ on untwisted RR ground states is modified by $(-1)^{\kappa_g}$, so the choice of~$\kappa$ indeed corresponds to a choice of a sign in the $G$-action on RR ground states.

In summary, in the $A_{G,\kappa}^c$-orbifold theory we precisely recover the results of \cite{IntriligatorVafa1990}:
\begin{align*}
\rho_{\text{(c,c)}} (g) \prod_{\Theta_i^h \in \Z}  x_i^{l_i} | 0 \rangle_{\text{(c,c)}}^h 
& = (-1)^{\kappa_g \kappa_h} \, \varepsilon (g,h) \, \frac{\det (g|_h)}{\det (g)} \, \E^{2\pi\I \sum_{\Theta_i^h \in \Z} \Theta_i^g l_i} \prod_{\Theta_i^h \in \Z} x_i^{l_i} | 0 \rangle_{\text{(c,c)}}^h \, , \\ 
\rho_{\text{RR}} (g) \prod_{\Theta_i^h \in \Z} x_i^{l_i} | 0 \rangle_{\text{RR}}^h
& = (-1)^{\kappa_g (\kappa_h+1)} \, \varepsilon (g,h) \, \det (g|_h) \, \E^{2\pi\I \sum_{\Theta_i^h \in \Z} \Theta_i^g l_i} \prod_{\Theta_i^h \in \Z} x_i^{l_i} | 0 \rangle_{\text{RR}}^h \, .
\end{align*}
\end{remark}

\subsection{Boundary sector}\label{subsec:boundarysec}

For a given orbifold defect~$A$, the boundary sector of the associated orbifold theory is given by the category of $A$-modules. In the special case of Landau-Ginzburg models with $A = A_G$ as in~\eqref{eq:AG} it was shown in \cite[Sect.\,7.1]{cr1210.6363} that these are precisely $G$-equivariant matrix factorisations, 
$$
\modu(A_G) \cong \hmf( \C[x], W )^G \, . 
$$

We shall now obtain the analogous result for discrete torsion orbifolds, where $A = \AGC$. By considering $\AGC$-modules we will see how we are naturally led to projective $G$-representations in our setting. More precisely, for $c \in H^2(G,U(1))$ we define the category $\hmf( \C[x], W )^{G,c}$ of \textsl{$c$-projective ($G$-equivariant) matrix factorisations} as follows. Objects are finite-rank matrix factorisations~$Q$ of~$W$, together with a set of isomorphisms $\{ \gamma_g : {}_g Q \rightarrow Q \}_{g \in G}$ such that $\gamma_e = 1$ and $\gamma_g \circ {}_g (\gamma_h) = c(g,h)\cdot \gamma_{gh}$ for all $g,h\in G$. If~$\gamma_h$ does not depend on the $x'$-variables (which is typically the case), then the latter condition simply reads
$$
\gamma_g \circ \gamma_h = c(g,h)\cdot \gamma_{gh} \, . 
$$
Thus for $c\neq 1$ the factorisation~$Q$ does not carry a linear, but a projective representation of~$G$. Morphisms $\Phi: Q \rightarrow P$ in $\hmf( \C[x], W )^{G,c}$ however only depend indirectly on~$c$: they are maps of matrix factorisations~$\Phi$ with the additional condition
$$
\Phi = \gamma_g^{(P)} \circ {}_g \Phi \circ (\gamma_g^{(Q)})^{-1} 
\quad \text{for all } g\in G. 
$$
Thus we have 
$$
\hmf( \C[x], W )^{G,c} = \hmf( \C[x], W )^{G} 
\quad \text{for } c = 1. 
$$

We can now show that boundary conditions in discrete torsion Landau-Ginzburg orbifolds are precisely $c$-projective matrix factorisations. 

\begin{theorem} \label{thm:modAGc}
$\modu(\AGC) \cong \hmf( \C[x], W )^{G,c}$. 
\end{theorem}

\begin{proof}
The argument is similar to the proof of \cite[Thm.\,7.2]{cr1210.6363}. Given a $c$-projective matrix factorisation $(Q, \{ \gamma_g \})$, we define $\rho_g = \gamma_g \circ {}_g (\lambda_Q)$ and $\rho = \sum_{g\in G} \rho_g : A_G \otimes Q \rightarrow Q$. Then $(Q,Ê\rho)$ is an $\AGC$-module iff for $C_{g,h}=1$ the diagram 
\be\label{eq:bigCdiagram}
\begin{tikzpicture}[
			     baseline=(current bounding box.base), 
			     >=stealth,
			     descr/.style={fill=white,inner sep=2.5pt}, 
			     normal line/.style={->}
			     ] 
\fill (1.75,0.95) circle (0pt) node (alpha) {{\small $\text{nat.}$}};
\fill (0,2.95) circle (0pt) node (alpha) {{\small $\text{def.}$}};
\fill (-1.75,0.0) circle (0pt) node (alpha) {{\small $D_1$}};
\fill (-4.25,0.0) circle (0pt) node (alpha) {{\small $\text{def.}$}};
\fill (2.5,-1.5) circle (0pt) node (alpha) {{\small $D_2$}};
\fill (0,-2.85) circle (0pt) node (alpha) {{\small $\text{def.}$}};
\fill (4.5,0.0) circle (0pt) node (alpha) {{\small $\text{def.}$}};
\matrix (m) [matrix of math nodes, row sep=3em, column sep=2.5em, text height=1.5ex, text depth=0.25ex] {%
{}_g I \otimes {}_h I \otimes Q &&&&&& {}_g I \otimes Q \\
&&& {}_g I \otimes {}_h Q &&& \\
&&&&& {}_g Q & \\
&&& {}_{gh} Q &&& \\
{}_{gh} I \otimes Q &&&&&& Q \\
};
\path[font=\scriptsize] (m-1-1) edge[->] node[auto] {$ 1_{{}_g I} \otimes \rho_h $} (m-1-7)
				  (m-1-1) edge[->] node[above, sloped] {$ 1_{{}_g I} \otimes {}_h (\lambda_Q) $} (m-2-4)
				  (m-1-1) edge[->] node[below, swap, sloped] {$ = {}_g (1_I \otimes {}_h(\lambda_Q)) $} (m-2-4)
				  (m-1-1) edge[->] node[below, sloped] {$ c(g,h) \cdot {}_g (\lambda_{{}_h I}) \otimes 1_Q $} (m-5-1)
				  (m-1-1) edge[->, out = -45, in=45] node[above, sloped] {$ c(g,h) \cdot {}_g (\lambda_{{}_h I} \otimes 1_Q )$} (m-5-1);
\path[font=\scriptsize] (m-2-4) edge[->] node[above, sloped] {$ 1_{{}_g I} \otimes \gamma_h $} (m-1-7)
				  (m-2-4) edge[->] node[below, sloped] {$ = {}_g (1_I \otimes \gamma_h) $} (m-1-7)
				  (m-2-4) edge[->] node[below, sloped] {$ {}_g (\lambda_{{}_h Q})$} (m-4-4);
\path[font=\scriptsize] (m-1-7) edge[->, out = -120, in=90] node[above, sloped] {$ {}_g (\lambda_Q) $} (m-3-6)
				  (m-1-7) edge[->] node[above, sloped] {$ \rho_g $} (m-5-7);
\path[font=\scriptsize] (m-3-6) edge[->, out = -90, in=120] node[above, sloped] {$ \gamma_g $} (m-5-7);
\path[font=\scriptsize] (m-4-4) edge[->] node[above, sloped] {$ {}_g (\gamma_h) $} (m-3-6)
				  (m-4-4) edge[->] node[above, sloped] {$ C_{g,h} \cdot \gamma_{gh} $} (m-5-7);
\path[font=\scriptsize] (m-5-1) edge[->] node[above, sloped] {$ {}_g ( {}_h (\lambda_Q)) $} (m-4-4)
				  (m-5-1) edge[->] node[above, sloped] {$ C_{g,h} \cdot \rho_{gh} $} (m-5-7);
\end{tikzpicture}
\ee
commutes. The subdiagrams indicated commute either by definition or by naturality of~$\lambda$. Thus~\eqref{eq:bigCdiagram} commutes iff the diagram obtained by setting the factors $c(g,h)$ on the two leftmost arrows to~1 while setting $C_{g,h} = c(g,h)$. This is true since the subdiagrams~$D_1$ and~$D_2$ then commute by naturality and by definition, respectively, and the leftmost and bottom subdiagrams continue to commute. The rest of the proof is exactly as that of \cite[Thm.\,7.2]{cr1210.6363}. 
\end{proof}

Again, the general theory of \cite{cr1210.6363, BCP1} applies. In particular it immediately follows that the Cardy condition 
in the form of \cite[Prop.\,3.16]{BCP1}
holds in $\hmf( \C[x], W )^{G,c}$: 
for all $Q,P \in \hmf( \C[x], W )^{G,c}$ and all maps $\Phi \in \End(Q)$, $\Psi \in \End(P)$ the overlap of the generalised boundary states $\beta^Q(\Phi)$ and $\beta^P(\Psi)$ equals the trace of over the operator ${}_\Psi m_{\Phi} = \Psi\circ (-) \circ \Phi$ on $\Hom(Q,P)$: 
$$
\big\langle \beta^Q(\Phi) , \beta^P(\Psi) \big\rangle = \operatorname{tr} ({}_\Psi m_{\Phi}) \, . 
$$
Furthermore, if $\AGC$ is symmetric the category $\hmf( \C[x], W )^{G,c}$ 
is Calabi-Yau with an explicitly known nondegenerate Serre pairing
$\langle -,- \rangle$; see \cite[Sect.\,3.3]{BCP1} for more details. 

\medskip

The question naturally arises what the relation between equivariant matrix factorisations with linear and with projective group actions is. In general they are not in one-to-one correspondence, but there are canonical ways to produce one from the other, as we explain in Example~\ref{ex:lineartoprojective} below and the text preceding it. This is easily done in terms of general defects to which we now turn.

\subsection{Defect sector}\label{subsec:defectsec}

A defect~$X$ between two orbifold theories built from two defects~$A$ and~$B$ is described by an $A$-$B$-bimodule. This means that there are two trivalent junction fields which give (unital) commuting left $A$- and right $B$-actions on~$X$, compatible with $A$- and $B$-multiplication: 
$$
\begin{tikzpicture}[very thick,scale=0.75,color=blue!50!black, baseline]
\draw (0,-1) node[below] (X) {{\small$X$}};
\draw[color=green!50!black] (-0.5,-1) node[below] (X) {{\small$A$}};
\draw[color=green!50!black] (-1,-1) node[below] (X) {{\small$A$}};
\draw (0,1) node[right] (Xu) {};
\draw (0,-1) -- (0,1); 
\fill[color=green!50!black] (0,-0.25) circle (2.9pt) node (meet) {};
\fill[color=green!50!black] (0,0.75) circle (2.9pt) node (meet2) {};
\draw[color=green!50!black] (-0.5,-1) .. controls +(0,0.25) and +(-0.25,-0.25) .. (0,-0.25);
\draw[color=green!50!black] (-1,-1) .. controls +(0,0.5) and +(-0.5,-0.5) .. (0,0.75);
\end{tikzpicture} 
=
\begin{tikzpicture}[very thick,scale=0.75,color=blue!50!black, baseline]
\draw (0,-1) node[below] (X) {{\small$X$}};
\draw[color=green!50!black] (-0.5,-1) node[below] (X) {{\small$A$}};
\draw[color=green!50!black] (-1,-1) node[below] (X) {{\small$A$}};
\draw (0,1) node[right] (Xu) {};
\draw (0,-1) -- (0,1); 
\fill[color=green!50!black] (0,0.75) circle (2.9pt) node (meet2) {};
\draw[-dot-, color=green!50!black] (-0.5,-1) .. controls +(0,1) and +(0,1) .. (-1,-1);
\draw[color=green!50!black] (-0.75,-0.2) .. controls +(0,0.5) and +(-0.5,-0.5) .. (0,0.75);
\end{tikzpicture} 
\, , \quad
\begin{tikzpicture}[very thick,scale=0.75,color=blue!50!black, baseline]
\draw (0,-1) node[below] (X) {{\small$X$}};
\draw[color=green!50!black] (0.5,-1) node[below] (X) {{\small$B$}};
\draw[color=green!50!black] (1,-1) node[below] (X) {{\small$B$}};
\draw (0,1) node[left] (Xu) {};
\draw (0,-1) -- (0,1); 
\fill[color=green!50!black] (0,0.75) circle (2.9pt) node (meet2) {};
\draw[-dot-, color=green!50!black] (0.5,-1) .. controls +(0,1) and +(0,1) .. (1,-1);
\draw[color=green!50!black] (0.75,-0.2) .. controls +(0,0.5) and +(0.5,-0.5) .. (0,0.75);
\end{tikzpicture} 
= 
\begin{tikzpicture}[very thick,scale=0.75,color=blue!50!black, baseline]
\draw (0,-1) node[below] (X) {{\small$X$}};
\draw[color=green!50!black] (0.5,-1) node[below] (X) {{\small$B$}};
\draw[color=green!50!black] (1,-1) node[below] (X) {{\small$B$}};
\draw (0,1) node[left] (Xu) {};
\draw (0,-1) -- (0,1); 
\fill[color=green!50!black] (0,-0.25) circle (2.9pt) node (meet) {};
\fill[color=green!50!black] (0,0.75) circle (2.9pt) node (meet2) {};
\draw[color=green!50!black] (0.5,-1) .. controls +(0,0.25) and +(0.25,-0.25) .. (0,-0.25);
\draw[color=green!50!black] (1,-1) .. controls +(0,0.5) and +(0.5,-0.5) .. (0,0.75);
\end{tikzpicture} 
, \quad
\begin{tikzpicture}[very thick,scale=0.75,color=blue!50!black, baseline]
\draw (0,-1) node[below] (X) {{\small$X$}};
\draw[color=green!50!black] (0.5,-1) node[below] (X) {{\small$B$}};
\draw[color=green!50!black] (-1,-1) node[below] (X) {{\small$A$}};
\draw (0,1) node[left] (Xu) {};
\draw (0.5,-1) node[right] (A) {};
\draw (-1,-1) node[left] (B) {};
\draw (0,-1) -- (0,1); 
\fill[color=green!50!black] (0,-0.3) circle (2.9pt) node (meet) {};
\fill[color=green!50!black] (0,0.3) circle (2.9pt) node (meet) {};
\draw[color=green!50!black] (-1,-1) .. controls +(0,0.5) and +(-0.5,-0.5) .. (0,0.3);
\draw[color=green!50!black] (0.5,-1) .. controls +(0,0.25) and +(0.5,-0.5) .. (0,-0.3);
\end{tikzpicture} 
\!\!
=
\!\!
\begin{tikzpicture}[very thick,scale=0.75,color=blue!50!black, baseline]
\draw (0,-1) node[below] (X) {{\small$X$}};
\draw[color=green!50!black] (-0.5,-1) node[below] (X) {{\small$A$}};
\draw[color=green!50!black] (1,-1) node[below] (X) {{\small$B$}};
\draw (0,1) node[right] (Xu) {};
\draw (1.0,-1) node[right] (A) {};
\draw (-0.5,-1) node[left] (B) {};
\draw (0,-1) -- (0,1); 
\fill[color=green!50!black] (0,-0.3) circle (2.9pt) node (meet) {};
\fill[color=green!50!black] (0,0.3) circle (2.9pt) node (meet) {};
\draw[color=green!50!black] (-0.5,-1) .. controls +(0,0.25) and +(-0.5,-0.5) .. (0,-0.3);
\draw[color=green!50!black] (1.0,-1) .. controls +(0,0.5) and +(0.5,-0.5) .. (0,0.3);
\end{tikzpicture} 
.
$$

In particular, for two parent Landau-Ginzburg models with potentials $W,W'$ and respective symmetry groups $G,G'$, a defect between the discrete torsion orbifolds associated to $c \in H^2(G,U(1))$, $c' \in H^2(G',U(1))$ is an $A_{G'}^{c'}$-$A_G^c$-bimodule. Analogously to the case of boundary conditions (which is in fact the special case $A_{G}^{c} = A_{\{ e \}}^1$ with $W=0$) one finds that $A_{G'}^{c'}$-$A_G^c$-bimodules are equivalent to $(G' \times G)$-equivariant matrix factorisations of $W'
-W$ with~$c'$- and $c$-projective representations of~$G'$ and~$G$, respectively. 

\medskip

According to \cite{ffrs0909.5013, cr1210.6363} defect fusion in orbifold theories is described by the tensor product over the intermediate algebra. Given three orbifolding defects $A,B,C$, an $A$-$B$-bimodule~$X$ and a $B$-$C$-bimodule~$Y$, the fusion of the corresponding defects is described by the $A$-$C$-bimodule $X \otimes_B Y$. Under mild assumptions (which are always satisfied for matrix factorisations) the latter can be computed via the projector 
\be\label{eq:tensorproj}
\pi = 
\begin{tikzpicture}[very thick,scale=0.75,color=blue!50!black, baseline]

\draw (-1,-1) node[left] (X) {{\small$X$}};
\draw (1,-1) node[right] (Xu) {{\small$Y$}};

\draw (-1,-1) -- (-1,1); 
\draw (1,-1) -- (1,1); 

\fill[color=green!50!black] (-1,0.6) circle (2.5pt) node (meet) {};
\fill[color=green!50!black] (1,0.6) circle (2.5pt) node (meet) {};

\draw[-dot-, color=green!50!black] (0.35,-0.0) .. controls +(0,-0.5) and +(0,-0.5) .. (-0.35,-0.0);

\draw[color=green!50!black] (0.35,-0.0) .. controls +(0,0.25) and +(-0.25,-0.25) .. (1,0.6);
\draw[color=green!50!black] (-0.35,-0.0) .. controls +(0,0.25) and +(0.25,-0.25) .. (-1,0.6);

\draw[color=green!50!black] (0,-0.75) node[Odot] (down) {}; 
\draw[color=green!50!black] (down) -- (0,-0.35); 

\draw[color=green!50!black] (0.1,0.55) node[left] (X) {{\small$B$}};

\end{tikzpicture} 
: X \otimes Y \lra X \otimes Y
\ee
and the splitting and projection maps $\xi : X \otimes_B Y \rightarrow X \otimes Y$ and $\vartheta: X \otimes Y \rightarrow X \otimes_B Y$ such that $\xi \vartheta = \pi$ and $\vartheta \xi = 1$. 
Here, $X\otimes Y \equiv X \otimes_{I_W} Y$ is the fusion in the unorbifolded theory, and junction fields on $X \otimes_B Y$ are obtained by the projection $\Phi \otimes \Psi \mapsto \Phi \otimes_B \Psi = \vartheta (\Phi \otimes \Psi) \xi$.

It follows that defect fusion can be explicitly computed in discrete torsion Landau-Ginzburg orbifolds, as we have concrete expressions for all constituent maps in~\eqref{eq:tensorproj}. 
Furthermore, as first shown in \cite{br0707.0922} the fusion $X\otimes Y$ in the unorbifolded theory is equivalent to a finite-rank matrix factorisation~$Z$, 
and thanks to the main result of \cite{dm1102.2957} implemented algorithmically in \cite{khovhompaper}, $Z$ can be determined explicitly. 

As a consequence of the above observations, the general results on compatibility of defect actions on bulk fields with both fusion and sphere correlators also hold in discrete torsion orbifolds. For details we refer to \cite[Sect.\,3.4]{BCP1}. 

\medskip

There are natural ways to construct equivariant matrix factorisations with projective group actions from ordinary (possibly equivariant) ones, and vice versa. 
To see this, let us consider a $G$-equivariant matrix factorisation $(Q,\{ \gamma_{g} \})$ of a Landau-Ginzburg potential~$W$. 
For another potential~$W'$ with symmetry group~$G'$ and $c' \in H^2(G',U(1))$ we may consider the associated orbifold theory with discrete torsion. 
A defect~$X$ between the two theories is an $A_{G'}^{c'}$-$A_G$-bimodule, 
and fusing with this defect gives us a functor 
\be\label{eq:lineartoprojective}
X \otimes_{A_G} (-): 
\hmf( \C[x], W )^{G} \lra \hmf( \C[x'], W' )^{G',c'}
\ee
from matrix factorisations $(Q, \{\gamma_g\})$ with linear $G$-representations to those with projective $G'$-representations. 

More precisely, the new $G'$-equivariant matrix factorisation is 
$$
\widehat Q := X\otimes_{A_{G}} Q
\quad \text{with projective group action} \quad
\widehat \gamma_{g'} := \Gamma_{g'} \otimes_{A_{G}} 1_{Q}
$$
where by definition the latter matrix is computed from the projective $G'$-action $\Gamma_{g'}$ on~$X$ (extracted from the left action of $A_{G'}^{c'}$ on~$X$) by tensoring with~$1_{Q}$ and then pre- and post-composing with the splitting maps of the projector in~\eqref{eq:tensorproj} with $B = A_{G}$ and $Y = Q$. 
Note that the projective action~$\widehat \gamma_{g'}$ on~$\widehat Q$ depends implicitly on the linear action $\gamma_{g}$ on~$Q$ as the latter is part of the projector in~\eqref{eq:tensorproj}. 

In some cases the projective group action~$\widehat \gamma_{g'}$ can be computed very explicitly: 

\begin{example}\label{ex:lineartoprojective}
\begin{enumerate}
\item 
We keep with the above notation but now assume $W=W'$ and that~$G$ is trivial, i.\,e.~we start with  a plain, non-equivariant matrix factorisation $Q \in \hmf(\C[x], W)$. Furthermore, we choose $X = A_{G'}^{c'}$ viewed as an $A_{G'}^{c'}$-$I_W$-bimodule. Then one finds that the map~\eqref{eq:lineartoprojective} sends~$Q$ to its $G'$-orbit $(\widehat Q, \widehat \gamma_{g'}) \in \hmf( \C[x], W)^{G',c'}$ with the regular $c'$-projective $G'$-representation: 
\be\label{eq:projaction1}
\widehat Q = \bigoplus _{g' \in G'} {}_{g'} Q \, , 
\quad 
\widehat \gamma_{g'} = \sum_{g'' \in G'} c'(g',g'') \cdot \pi_{g'',g'g''}
\ee
where $\pi_{g'',g'g''} : {}_{g'} ({}_{g''} Q) \cong {}_{g'g''} Q$ simply permutes the summands of~$\widehat Q$. 
\item 
We keep $W=W'$ but relax the assumption on~$G$ such that $H^2(G,U(1)) = 0$ should hold, e.\,g.~$G= \Z_d$. Furthermore, we take~$G'$ to be of the form $H \times G$ for any symmetry group~$H$ of~$W$. By our assumption on~$G$ we have that $A_G^{c'} \cong A_G$ as an algebra for any $c' \in H^2( H \times G, U(1))$, so we can set $X = A_{H \times G}^{c'} \otimes_{A_G^{c'}} A_G \cong A_{H \times G}^{c'}$ as an $A_{H \times G}^{c'}$-$A_G$-bimodule. 
This maps $(Q, \{\gamma_g\}) \in \hmf(\C[x], W)^G$ to the induced $c'$-projective matrix factorisation 
\be\label{eq:projaction2}
\widehat Q = \bigoplus _{k \in H} {}_{k} Q \, , 
\quad 
\widehat \gamma_{(h,g)} = \sum_{k \in H} \frac{c'(g,hk)}{c'(hk,g) \, c'(g,h)} \,  {}_{hk} (\gamma_g) \circ {}_g(\pi_{k,hk})
\ee
in $\hmf( \C[x], W )^{G',c'}$, where $c'(g,h)$ is short for $c'\big( (e_H, g), (h, e_G) \big)$ etc. 
\end{enumerate}
Note that one can also directly verify that~\eqref{eq:projaction1} and~\eqref{eq:projaction2} furnish projective representations, but this is not necessary as it follows from the general defect discussion. 
\end{example}

In conclusion, we observe that \textsl{discrete torsion Landau-Ginzburg orbifolds are (generalised) orbifolds of Landau-Ginzburg orbifolds, and vice versa}. In fact, this result applies to a much larger class of two-dimensional quantum field theories,
and in Theorem~\ref{thm:orbiequi} below we prove this equivalence in an abstract setting that in particular covers all topological field theories with defects.\footnote{For rational conformal field theories an even stronger statement is true, as was first explained in \cite{ffrs0909.5013}, namely (under certain assumptions on the symmetry algebras, uniqueness of the vacuum and nondegeneracy of two-point functions) all theories of fixed central charge are orbifold equivalent.} 

In the case of Landau-Ginzburg models there are two implications of this general result that are worth mentioning. One is that 
$$
\hmf(\C[x], W)^{G,c} \cong \modu\! \big( A_G \otimes A_G^c \otimes A_G \big) 
$$
where the right-hand side only involves linear $G$-actions: $c$-projective matrix factorisations in $\hmf(\C[x], W)^{G,c}$ can equivalently be described as objects in $\hmf(\C[x], W)^{G}$ which come with the extra structure of being a module over the orbifolding defect $A_G \otimes A_G^c \otimes A_G$. 
This is a manifestation of the general fact that anything in a theory with discrete torsion can be described within the theory without discrete torsion. 

Another consequence of Theorem~\ref{thm:orbiequi} is that any unorbifolded Landau-Ginzburg model with symmetry group~$G$ is a (generalised) orbifold of its $G$-orbifold by the \textsl{quantum symmetry defect} 
$$
\mathcal A_G^{\text{qs}} = A_G \otimes A_G
$$
viewed as an equivariant matrix factorisation. In particular this implies
\be\label{eq:quantum-symmetry-orbifold}
\hmf(\C[x], W) \cong \modu ( 
\mathcal
A_G^{\text{qs}} ) \, . 
\ee
For abelian~$G$, results like this are well-known in 
the conventional approach to orbifolds in the CFT literature 
(see e.\,g.~\cite{ginsparg}). 
The corresponding orbifold defect was first considered in~\cite{br0712.0188}, and in our setting it can be written in terms of Nakayama twists, e.\,g.~
$$
\mathcal A_G^{\text{qs}} \cong \bigoplus_{h\in G} {}_{\gamma_{A_G}^h} \!\!(A_G)
\quad \text{for} \quad 
G = \Z_d
$$ 
(see Appendix~\ref{app:qsdefect} for details, and~\eqref{eq:Nakayama} for the definition of~$\gamma_{A_G}$).
We stress however that these results hold true even for nonabelian~$G$, and it is the defect $\mathcal A_G^{\text{qs}}$ that replaces the role of the quantum symmetry group in that case. 
Put differently, orbifolding by the \textsl{defect} $\mathcal A_G^{\text{qs}}$ may have nothing to do with an orbifold \textsl{group}: for nonabelian~$G$ there may not be any symmetry group~$H$ such that $\mathcal A_G^{\text{qs}}$ is isomorphic to $A_H$.

\section{Disc correlators}\label{sec:disccorrelators}
 
In this section we study properties of disc correlators which in particular simplify the computation of RR-charges. The general discussion is valid for any orbifold theory; we apply it to Landau-Ginzburg models in Example~\ref{ex:inducedcharges} and provide the corresponding CFT perspective in Appendix~\ref{app:CFT}. 

\subsection{Invariance and selection rule}

Correlators in orbifold theories based on a symmetry group $G$ are computed from correlators in the parent theory with twisted field insertions. Since the vacuum of the parent theory is invariant under $G$, the correlators must also be $G$-invariant. This condition has to be satisfied for both projected and unprojected twisted field insertions, and it provides a useful selection rule for the correlators. One can formulate an analogous notion of invariance also for defect orbifolds, which we discuss in the following for the case of disc correlators. 

We consider an orbifold theory constructed from a defect $A$, with a boundary condition given by an $A$-module $Q$. The disc correlator with bulk and boundary insertions $\alpha$ and $\psi$, respectively, is given by the diagram
\be\label{eq:orbiDisc}
\begin{tikzpicture}[very thick,scale=0.75,color=blue!50!black, baseline]
\nicepalecolourscheme (0,0) circle (1.5);
\draw (0,0) circle (1.5);
\fill (-45:1.55) circle (0pt) node[right] {{\small$Q$}};
\draw[<-, very thick] (0.100,-1.5) -- (-0.101,-1.5) node[above] {}; 
\draw[<-, very thick] (-0.100,1.5) -- (0.101,1.5) node[below] {}; 
\fill[color=green!50!black] (135:0) circle (2.5pt) node[right] {{\small$\alpha$}};
\fill (0:1.5) circle (2.5pt) node[left] {{\small$\psi$}};
\draw[color=green!50!black] (0,0) .. controls +(0,0.6) and +(-0.4,-0.4) .. (45:1.5);
\fill[color=green!50!black] (45:1.5) circle (2.5pt) node[right] {{}};
\end{tikzpicture} 
\equiv
\begin{tikzpicture}[very thick,scale=0.75,color=blue!50!black, baseline]
\nicepalecolourscheme (0,0) circle (1.5);
\draw (0,0) circle (1.5);
\fill (-45:1.55) circle (0pt) node[right] {{\small$Q$}};
\draw[<-, very thick] (0.100,-1.5) -- (-0.101,-1.5) node[above] {}; 
\draw[<-, very thick] (-0.100,1.5) -- (0.101,1.5) node[below] {}; 
\fill (0:1.5) circle (2.5pt) node[left] {{\small$\psi$}};
\fill (45:1.5) circle (2.5pt) node[left] {{\small$\beta_Q(\alpha)$}};
\end{tikzpicture} 
= \big\langle \beta_Q(\alpha) \circ \psi \big\rangle_Q
\ee
where $\beta_Q$ denotes the bulk-boundary map, see e.\,g.~\cite[Sect.\,3.3]{BCP1}. 
This correlator satisfies the property \cite[Prop.\,4.6]{cr1210.6363}
\be\label{eq:discinvariance}
\begin{tikzpicture}[very thick,scale=0.75,color=blue!50!black, baseline]
\nicepalecolourscheme (0,0) circle (1.5);
\draw (0,0) circle (1.5);
\fill (-45:1.55) circle (0pt) node[right] {{\small$Q$}};
\draw[<-, very thick] (0.100,-1.5) -- (-0.101,-1.5) node[above] {}; 
\draw[<-, very thick] (-0.100,1.5) -- (0.101,1.5) node[below] {}; 
\fill[color=green!50!black] (135:0) circle (2.5pt) node[right] {{\small$\alpha$}};
\fill (0:1.5) circle (2.5pt) node[left] {{\small$\psi$}};
\draw[color=green!50!black] (0,0) .. controls +(0,0.6) and +(-0.4,-0.4) .. (45:1.5);
\fill[color=green!50!black] (45:1.5) circle (2.5pt) node[right] {{}};
\end{tikzpicture} 
\!\!\! =
\begin{tikzpicture}[very thick,scale=0.75,color=blue!50!black, baseline]
\nicepalecolourscheme (0,0) circle (1.5);
\draw (0,0) circle (1.5);
\fill (-45:1.55) circle (0pt) node[right] {{\small$Q$}};
\draw[<-, very thick] (0.100,-1.5) -- (-0.101,-1.5) node[above] {}; 
\draw[<-, very thick] (-0.100,1.5) -- (0.101,1.5) node[below] {}; 
\fill[color=green!50!black] (135:0) circle (2.5pt) node[right] {{\small$\alpha$}};
\fill[color=green!50!black] (120:0.45) circle (2.5pt) node[left] {{\small$\gamma_A$}};
\fill (0:1.5) circle (2.5pt) node[left] {{\small$\psi$}};
\draw[color=green!50!black] (0,0) .. controls +(0,0.6) and +(-0.4,-0.4) .. (45:1.5);
\fill[color=green!50!black] (45:1.5) circle (2.5pt) node[right] {{}};
\draw[color=green!50!black] (-0.5,-0.6) .. controls +(0,0.6) and +(-0.8,-0.3) .. (60:1.5);
\fill[color=green!50!black] (60:1.5) circle (2.5pt) node[right] {{}};
\draw[color=green!50!black] (0.0,-0.6) .. controls +(0,0.3) and +(-0.2,0.0) .. (-10:1.5);
\fill[color=green!50!black] (-10:1.5) circle (2.5pt) node[right] {{}};
\draw[-dot-, color=green!50!black] (0.0,-0.6) .. controls +(0,-0.5) and +(0,-0.5) .. (-0.5,-0.6);
\draw[color=green!50!black] (-0.25,-1.25) node[Odot] (down) {}; 
\draw[color=green!50!black] (down) -- (-0.25,-0.95); 
\end{tikzpicture} 
\equiv \big\langle {}_{\gamma_A}P (\beta_Q(\alpha) \circ \psi) \big\rangle_Q
\ee
with ${}_{\gamma_A}P$ being the boundary orbifold projector twisted by the Nakayama automorphism~\eqref{eq:Nakayama}. This relation should be interpreted as a selection rule for the disc correlator, to wit only the ${}_{\gamma_A}P$-invariant part of $\beta_Q(\alpha) \circ \psi$ contributes.

Note that in the case $A=A_G^c$, \eqref{eq:discinvariance} becomes the condition of $G$-invariance of the correlator. Indeed, letting $\alpha=\phi_h$ be an RR ground state in the $h$-twisted sector and writing out $A_G^c$ in components, we obtain
\begin{align*}
\big\langle \beta_Q(\phi_h) \circ \psi \big\rangle_Q &=
\begin{tikzpicture}[very thick,scale=0.75,color=blue!50!black, baseline]
\nicepalecolourscheme (0,0) circle (1.5);
\draw (0,0) circle (1.5);
\fill (-45:1.55) circle (0pt) node[right] {{\small$Q$}};
\draw[<-, very thick] (0.100,-1.5) -- (-0.101,-1.5) node[above] {}; 
\draw[<-, very thick] (-0.100,1.5) -- (0.101,1.5) node[below] {}; 
\fill[color=green!65!black] (135:0) circle (2.5pt) node[right] {{\small$\phi_h$}};
\fill[color=green!65!black] (49:1) circle (0pt) node[right] {{\tiny$h$}};
\fill (0:1.5) circle (2.5pt) node[left] {{\small$\psi$}};
\draw[color=green!65!black] (0,0) .. controls +(0,0.6) and +(-0.4,-0.4) .. (45:1.5);
\fill[color=green!65!black] (45:1.5) circle (2.5pt) node[right] {{}};
\end{tikzpicture} 
\!\!\!= \sum_{g\in G}
\begin{tikzpicture}[very thick,scale=0.75,color=blue!50!black, baseline]
\nicepalecolourscheme (0,0) circle (1.5);
\draw (0,0) circle (1.5);
\fill (-45:1.55) circle (0pt) node[right] {{\small$Q$}};
\draw[<-, very thick] (0.100,-1.5) -- (-0.101,-1.5) node[above] {}; 
\draw[<-, very thick] (-0.100,1.5) -- (0.101,1.5) node[below] {}; 
\fill[color=green!65!black] (135:0) circle (2.5pt) node[right] {{\small$\phi_h$}};
\fill[color=green!65!black] (49:1) circle (0pt) node[right] {{\tiny$h$}};
\fill[color=green!65!black] (120:0.45) circle (2.5pt) node[left] {{\small$\gamma_{{}_g I}$}};
\fill[color=green!65!black] (-0.35,-0.75) circle (0pt) node[left] {{\tiny$g$}};
\fill (0:1.5) circle (2.5pt) node[left] {{\small$\psi$}};
\draw[color=green!65!black] (0,0) .. controls +(0,0.6) and +(-0.4,-0.4) .. (45:1.5);
\fill[color=green!65!black] (45:1.5) circle (2.5pt) node[right] {{}};
\draw[color=green!65!black] (-0.5,-0.6) .. controls +(0,0.6) and +(-0.8,-0.3) .. (60:1.5);
\fill[color=green!65!black] (60:1.5) circle (2.5pt) node[right] {{}};
\draw[color=green!65!black] (0.0,-0.6) .. controls +(0,0.3) and +(-0.2,0.0) .. (-10:1.5);
\fill[color=green!65!black] (-10:1.5) circle (2.5pt) node[right] {{}};
\draw[-dot-, color=green!65!black] (0.0,-0.6) .. controls +(0,-0.5) and +(0,-0.5) .. (-0.5,-0.6);
\draw[color=green!65!black] (-0.25,-1.25) node[Odot] (down) {}; 
\draw[color=green!65!black] (down) -- (-0.25,-0.95); 
\end{tikzpicture} 
\!\!\! = |G| \sum_{g\in G}
\begin{tikzpicture}[very thick,scale=0.75,color=blue!50!black, baseline]
\nicepalecolourscheme (0,0) circle (1.5);
\draw (0,0) circle (1.5);
\fill (-45:1.55) circle (0pt) node[right] {{\small$Q$}};
\draw[<-, very thick] (0.100,-1.5) -- (-0.101,-1.5) node[above] {}; 
\draw[<-, very thick] (-0.100,1.5) -- (0.101,1.5) node[below] {}; 
\fill[color=green!65!black] (0,0) circle (2.5pt) node[left] {{}};
\fill[color=green!65!black] (0.2,0) circle (0pt) node[left] {{\small$\phi_h$}};
\fill[color=green!65!black] (160:0.58) circle (2.5pt) node[left] {{}};
\fill[color=green!65!black] (160:0.5) circle (0pt) node[left] {{\small$\gamma_{{}_g I}$}};
\fill (0:1.5) circle (2.5pt) node[left] {{\small$\psi$}};
\draw[color=green!65!black] (0,0) .. controls +(0,0.6) and +(-0.4,-0.4) .. (45:1.5);
\fill[color=green!65!black] (45:1.5) circle (2.5pt) node[right] {{}};
\draw[color=green!65!black] (-0.6,-0.3) .. controls +(0,0.6) and +(-1,-0.3) .. (59:0.7);
\fill[color=green!65!black] (59:0.7) circle (2.5pt) node[right] {{}};
\draw[color=green!65!black] (0.4,-0.3) .. controls +(0,0.2) and +(0.2,-0.1) .. (78:0.3);
\fill[color=green!65!black] (78:0.3) circle (2.5pt) node[right] {{}};
\draw[-dot-, color=green!65!black] (0.4,-0.3) .. controls +(0,-0.5) and +(0,-0.5) .. (-0.6,-0.3);
\draw[color=green!65!black] (-0.1,-0.95) node[Odot] (down) {}; 
\draw[color=green!65!black] (down) -- (-0.1,-0.75); 
\fill[color=green!65!black] (-0.25,-0.7) circle (0pt) node[left] {{\tiny$g$}};
\draw[color=green!65!black] (0.75,-0.5) .. controls +(0,0.6) and +(-0.3,-0.1) .. (20:1.5);
\fill[color=green!65!black] (20:1.5) circle (2.5pt) node[right] {{}};
\draw[color=green!65!black] (1.25,-0.5) .. controls +(0,0.2) and +(-0.1,0.0) .. (-10:1.5);
\fill[color=green!65!black] (-10:1.5) circle (2.5pt) node[right] {{}};
\draw[-dot-, color=green!65!black] (0.75,-0.5) .. controls +(0,-0.25) and +(0,-0.25) .. (1.25,-0.5);
\draw[color=green!65!black] (1,-0.95) node[Odot] (down) {}; 
\draw[color=green!65!black] (down) -- (1,-0.75); 
\fill[color=green!65!black] (0.95,-0.7) circle (0pt) node[left] {{\tiny$g$}};
\fill[color=green!65!black] (1.15,0.95) circle (0pt) node[left] {{\tiny$ghg^{-1}$}};
\end{tikzpicture} 
\\
&= \frac{1}{|G|} \sum_{g\in G}
\big\langle \beta_Q (\rho_{\text{RR}}(g)(\phi_h)) \circ \gamma_g \circ {}_g \psi \circ \gamma_g^{-1} \big\rangle_Q
\end{align*}
where $\gamma_{{}_g I} = \gamma_{A_G^c}|_{{}_g I}$, and in the third step we used the Frobenius and module properties. 

In the case $\psi=1_Q$ and $\alpha$ being an RR ground state,~\eqref{eq:orbiDisc} computes the RR charge of $Q$ under $\alpha$. The property 
\be
{}_{\gamma_A} P(\beta_Q(\alpha)) 
=\!\! 
\begin{tikzpicture}[very thick,scale=0.75,color=blue!50!black, baseline]
\draw (0,-1) node[left] (X) {};
\draw (0,1) node[left] (Xu) {};
\draw (0,-1) -- (0,1); 
\fill[color=green!50!black] (-0.9,-0.4) circle (0pt) node[right] {{\small$\alpha$}};
\fill[color=green!50!black] (-0.8,-0.4) circle (2.5pt) node[right] {{}};
\fill[color=green!50!black]  (0,0.4) circle (2.5pt) node (meet) {};
\draw[color=green!50!black]  (-0.8,-0.4) .. controls +(0,0.4) and +(-0.3,-0.1) .. (0,0.4);
\draw[-dot-, color=green!50!black] (-0.3,-0.6) .. controls +(0,-0.5) and +(0,-0.5) .. (-1.4,-0.6);
\draw[color=green!50!black] (-0.3,-0.6) .. controls +(0,0.2) and +(0,-0.1) .. (0,-0.3);
\draw[color=green!50!black] (-1.4,-0.6) .. controls +(0,0.6) and +(-0.3,-0.1) .. (0,0.7);
\fill[color=green!50!black] (0,-0.3) circle (2.5pt) node (down) {};
\fill[color=green!50!black] (0,0.7) circle (2.5pt) node (up) {};
\draw[color=green!50!black] (-0.85,-1.3) node[Odot] (unit) {}; 
\draw[color=green!50!black] (-0.85,-1.0) -- (unit);
\fill[color=green!50!black] (-1.35,-0.4) circle (2.5pt) node[left] {{\small$\gamma_A$}};
\end{tikzpicture} 
=\!\!
\begin{tikzpicture}[very thick,scale=0.75,color=blue!50!black, baseline]
\draw (0,-1) node[left] (X) {};
\draw (0,1) node[left] (Xu) {};
\draw (0,-1) -- (0,1); 
\fill[color=green!50!black] (-0.9,-0.4) circle (0pt) node[right] {{\small$\alpha$}};
\fill[color=green!50!black] (-0.8,-0.4) circle (2.5pt) node[right] {{}};
\fill[color=green!50!black]  (0,0.4) circle (2.5pt) node (meet) {};
\draw[color=green!50!black]  (-0.8,-0.4) .. controls +(0,0.4) and +(-0.3,-0.1) .. (0,0.4);
\draw[-dot-, color=green!50!black] (-0.3,-0.6) .. controls +(0,-0.5) and +(0,-0.5) .. (-1.4,-0.6);
\draw[color=green!50!black] (-0.3,-0.6) .. controls +(0,0.2) and +(0.3,-0.1) .. (-0.625,0);
\draw[color=green!50!black] (-1.4,-0.6) .. controls +(0,0.6) and +(-0.8,-0.05) .. (-0.325,0.25);
\fill[color=green!50!black] (-0.625,0) circle (2.5pt) node (down) {};
\fill[color=green!50!black] (-0.325,0.25) circle (2.5pt) node (up) {};
\draw[color=green!50!black] (-0.85,-1.3) node[Odot] (unit) {}; 
\draw[color=green!50!black] (-0.85,-1.0) -- (unit);
\fill[color=green!50!black] (-1.365,-0.4) circle (2.5pt) node[left] {{\small$\gamma_A$}};
\end{tikzpicture} 
=
\beta_Q(\pirr(\alpha))
\ee
together with~\eqref{eq:discinvariance} then implies that only $\pirr (\alpha)$ contributes (where $\pirr$ is the projector to RR ground states, see \cite[Sect.\,3.2]{BCP1}). 
In particular, this shows that $A_G^c$-modules are charged only under those RR ground states that are invariant in the $A_G^c$-orbifold theory. 

\subsection{Induced modules}\label{subsec:inducedmodules}

The computation of D-brane RR charges simplifies for induced modules. In the construction of an induced module, one starts with an arbitrary $A'$-module $Q$ and builds the $A$-module $\widehat Q=A \otimes_{A'} Q$ where~$A'$ is another separable Frobenius algebra together with an algebra map $\varphi: A' \rightarrow A$, i.\,e. 
$$
\begin{tikzpicture}[very thick,scale=0.75,color=green!50!black, baseline=0.4cm]
\draw[-dot-] (3,0) .. controls +(0,1) and +(0,1) .. (2,0);
\draw (2.5,0.75) -- (2.5,1.5); 
\fill (2,0) circle (2.5pt) node[left] (alpha) {{\small $\varphi$}};
\fill (3,0) circle (2.5pt) node[right] (beta) {{\small $\varphi$}};
\draw (2,0) -- (2,-0.5); 
\draw (3,0) -- (3,-0.5); 
\fill (2,-0.5) circle (0pt) node[left] {{\small $A'$}};
\fill (3,-0.5) circle (0pt) node[right] {{\small $A'$}};
\fill (2.5,1.5) circle (0pt) node[right] {{\small $A$}};
\end{tikzpicture} 
=
\begin{tikzpicture}[very thick,scale=0.75,color=green!50!black, baseline=0.4cm]
\draw[-dot-] (3,0) .. controls +(0,1) and +(0,1) .. (2,0);
\draw (2.5,0.75) -- (2.5,1.5); 
\fill (2.5,1.15) circle (2.5pt) node[left] (alpha) {{\small $\varphi$}};
\draw (2,0) -- (2,-0.5); 
\draw (3,0) -- (3,-0.5); 
\fill (2,-0.5) circle (0pt) node[left] {{\small $A'$}};
\fill (3,-0.5) circle (0pt) node[right] {{\small $A'$}};
\fill (2.5,1.5) circle (0pt) node[right] {{\small $A$}};
\end{tikzpicture} 
\, , \quad
\begin{tikzpicture}[very thick,scale=0.75,color=green!50!black, baseline=0.4cm]
\draw (0,-0.4) node[Odot] (D) {}; 
\draw (D) -- (0,1.5); 
\fill (0,0.65) circle (2.5pt) node[left] (alpha) {{\small $\varphi$}};
\fill (0,1.5) circle (0pt) node[right] {{\small $A$}};
\end{tikzpicture} 
=
\begin{tikzpicture}[very thick,scale=0.75,color=green!50!black, baseline=0.4cm]
\draw (0,-0.4) node[Odot] (D) {}; 
\draw (D) -- (0,1.5); 
\fill (0,1.5) circle (0pt) node[right] {{\small $A$}};
\end{tikzpicture} 
\, . 
$$
The algebra $A$ is then in particular a right $A'$-module via
$$
\begin{tikzpicture}[very thick,scale=0.75,color=green!50!black, baseline]
\draw (0,-1) node[left] (X) {{\small$A\vphantom{A'}$}};
\draw (0,1) node[left] (Xu) {{\small$A$}};
\draw (0.5,-1) node[right] {{\small$A'$}};
\draw (0,-1) -- (0,1); 
\fill (0.45,-0.5) circle (2.5pt) node[right] {{\small $\varphi$}};
\fill (0,0.0) circle (2.5pt) node (meet) {};
\draw(0.5,-1) .. controls +(0,0.5) and +(0.5,-0.5) .. (0,0.0);
\end{tikzpicture} 
: A \otimes A' \lra A \, .
$$
For $\widehat Q$ of this form and $\psi=1_Q$, 
we express the correlator~\eqref{eq:orbiDisc} in terms of $Q$ as
$$
\begin{tikzpicture}[very thick,scale=1,color=blue!50!black, baseline]
\nicepalecolourscheme (0,0) circle (1.5);
\draw (0,0) circle (1.5);
\fill (-45:1.55) circle (0pt) node[right] {{\small$\widehat Q$}};
\draw[<-, very thick] (0.100,-1.5) -- (-0.101,-1.5) node[above] {}; 
\draw[<-, very thick] (-0.100,1.5) -- (0.101,1.5) node[below] {}; 
\fill[color=green!50!black] (135:0) circle (1.75pt) node[right] {{\small$\alpha$}};
\draw[color=green!50!black] (0,0) .. controls +(0,0.6) and +(-0.4,-0.4) .. (45:1.5);
\fill[color=green!50!black] (45:1.5) circle (1.75pt) node[right] {{}};
\fill[color=green!50!black] (68:0.5) circle (0pt) node[right] {{\tiny$A$}};
\end{tikzpicture} 
\!\!\!=
\begin{tikzpicture}[very thick,scale=1,color=blue!50!black, baseline]
\nicepalecolourscheme (0,0) circle (1.5);
\draw (0,0) circle (1.5);
\draw[color=green!50!black] (0,0) circle (0.4);
\draw[<-, color=green!50!black, very thick] (0.100,-0.4) -- (-0.101,-0.4) node[above] {}; 
\draw[<-, color=green!50!black,  very thick] (-0.100,0.4) -- (0.101,0.4) node[below] {}; 
\fill (-45:1.55) circle (0pt) node[right] {{\small$Q$}};
\draw[<-, very thick] (0.100,-1.5) -- (-0.101,-1.5) node[above] {}; 
\draw[<-, very thick] (-0.100,1.5) -- (0.101,1.5) node[below] {}; 
\fill[color=green!50!black] (0.08,0) circle (0pt) node[left] {{\small$\alpha$}};
\fill[color=green!50!black] (135:0) circle (1.75pt) node[left] {{}};
\draw[color=green!50!black] (0,0) .. controls +(0,0.15) and +(-0.15,-0.15) .. (30:0.4);
\fill[color=green!50!black] (30:0.4) circle (1.75pt) node[right] {{}};
\draw[color=green!50!black] (55:0.4) .. controls +(0.5,-0.05) and +(0,0.1) .. (0.65,0);
\fill[color=green!50!black] (55:0.4) circle (1.75pt) node[right] {{}};
\draw[color=green!50!black] (1.3,0) .. controls +(0,0.1) and +(0,0) .. (30:1.5);
\fill[color=green!50!black] (30:1.5) circle (1.75pt) node[right] {{}};
\draw[-dot-, color=green!50!black] (0.65,0) .. controls +(0,-0.5) and +(0,-0.5) .. (1.3,0);
\draw[color=green!50!black] (0.975,-0.65) node[Odot] (down) {}; 
\draw[color=green!50!black] (down) -- (0.975,-0.4); 
\fill[color=green!50!black] (0.65,0) circle (1.75pt) node[right] {{\small$\varphi$}};
\fill[color=green!50!black] (300:0.35) circle (0pt) node[right] {{\tiny$A$}};
\fill[color=green!50!black] (0.925,-0.45) circle (0pt) node[right] {{\tiny$A'$}};
\end{tikzpicture} 
\!\!\!=
\begin{tikzpicture}[very thick,scale=1,color=blue!50!black, baseline]
\nicepalecolourscheme (0,0) circle (1.5);
\draw (0,0) circle (1.5);
\draw[color=green!50!black] (0,0) circle (0.4);
\draw[<-, color=green!50!black, very thick] (0.100,-0.4) -- (-0.101,-0.4) node[above] {}; 
\draw[<-, color=green!50!black,  very thick] (-0.100,0.4) -- (0.101,0.4) node[below] {}; 
\fill (-45:1.55) circle (0pt) node[right] {{\small$Q$}};
\draw[<-, very thick] (0.100,-1.5) -- (-0.101,-1.5) node[above] {}; 
\draw[<-, very thick] (-0.100,1.5) -- (0.101,1.5) node[below] {}; 
\fill[color=green!50!black] (0.08,0) circle (0pt) node[left] {{\small$\alpha$}};
\fill[color=green!50!black] (135:0) circle (1.75pt) node[left] {{}};
\draw[color=green!50!black] (0,0) .. controls +(0,0.15) and +(-0.15,-0.15) .. (30:0.4);
\fill[color=green!50!black] (30:0.4) circle (1.75pt) node[right] {{}};
\draw[color=green!50!black] (0.65,0) -- (27:0.728);
\draw[-dot-, color=green!50!black] (27:0.728) .. controls +(0,0.56) and +(0,0.56) .. (55:0.4);
\draw[color=green!50!black] (0.44,1.05) node[Odot] (down) {}; 
\draw[color=green!50!black] (down) -- (0.44,0.8); 
\fill[color=green!50!black] (55:0.4) circle (1.75pt) node[right] {{}};
\draw[color=green!50!black] (1.3,0) .. controls +(0,0.1) and +(0,0) .. (30:1.5);
\fill[color=green!50!black] (30:1.5) circle (1.75pt) node[right] {{}};
\draw[-dot-, color=green!50!black] (0.65,0) .. controls +(0,-0.5) and +(0,-0.5) .. (1.3,0);
\draw[color=green!50!black] (0.975,-0.65) node[Odot] (down) {}; 
\draw[color=green!50!black] (down) -- (0.975,-0.4); 
\fill[color=green!50!black] (0.65,0) circle (1.75pt) node[right] {{\small$\varphi$}};
\fill[color=green!50!black] (300:0.35) circle (0pt) node[right] {{\tiny$A$}};
\fill[color=green!50!black] (0.925,-0.45) circle (0pt) node[right] {{\tiny$A'$}};
\end{tikzpicture} 
$$
where in the first step we employed the projector~\eqref{eq:tensorproj} and in the second step we used the Frobenius property for $A$. Recognising the $\pirr$ projector as \cite[Lem.\,3.7]{BCP1}
$$
\pirr(\alpha)=
\begin{tikzpicture}[very thick,scale=1,color=green!50!black, baseline]
\draw[color=green!50!black] (0,0) circle (0.4);
\draw[<-, very thick] (0.100,-0.4) -- (-0.101,-0.4) node[above] {}; 
\draw[<-, very thick] (-0.100,0.4) -- (0.101,0.4) node[below] {}; 
\fill (0.08,0) circle (0pt) node[left] {{\small$\alpha$}};
\fill (135:0) circle (1.75pt) node[left] {{}};
\draw (0,0) .. controls +(0,0.15) and +(-0.15,-0.15) .. (30:0.4);
\fill (30:0.4) circle (1.75pt) node[right] {{}};
\fill (55:0.4) circle (1.75pt) node[right] {{}};
\draw (0.23,0.32) -- (0.23,0.8);
\end{tikzpicture} 
$$
and assuming that $\pirr(\alpha)=\alpha$, the correlator further simplifies to
$$
\begin{tikzpicture}[very thick,scale=0.75,color=blue!50!black, baseline]
\nicepalecolourscheme (0,0) circle (1.5);
\draw (0,0) circle (1.5);
\fill (-45:1.55) circle (0pt) node[right] {{\small$Q$}};
\draw[<-, very thick] (0.100,-1.5) -- (-0.101,-1.5) node[above] {}; 
\draw[<-, very thick] (-0.100,1.5) -- (0.101,1.5) node[below] {}; 
\fill[color=green!50!black] (0,0) circle (2.5pt) node[left] {{\small$\alpha$}};
\draw[-dot-, color=green!50!black] (0,0) .. controls +(0,0.5) and +(0,0.5) .. (0.65,0);
\draw[color=green!50!black] (0.325,0.65) node[Odot] (down) {}; 
\draw[color=green!50!black] (down) -- (0.325,0.4); 
\draw[color=green!50!black] (1.3,0) .. controls +(0,0.1) and +(0,0) .. (30:1.5);
\fill[color=green!50!black] (30:1.5) circle (2.5pt) node[right] {{}};
\draw[-dot-, color=green!50!black] (0.65,0) .. controls +(0,-0.5) and +(0,-0.5) .. (1.3,0);
\draw[color=green!50!black] (0.975,-0.65) node[Odot] (down) {}; 
\draw[color=green!50!black] (down) -- (0.975,-0.4); 
\fill[color=green!50!black] (0.65,0) circle (2.5pt) node[right] {{\small$\varphi$}};
\fill[color=green!50!black] (0.25,0.35) circle (0pt) node[left] {{\tiny$A$}};
\fill[color=green!50!black] (0.85,-0.4) circle (0pt) node[right] {{\tiny$A'$}};
\end{tikzpicture} 
.
$$

In the special case $A=A_G^c$ and $A'=A_H^c$, with $H$ a subgroup of $G$, we have an embedding $\varphi: A_H^c \rightarrow A_G^c$, and only $\alpha \in \Hom(I,A_H^c)$ gives a nonvanishing contribution to the correlator,\footnote{If $H$ is the trivial subgroup, this is the standard statement that a `bulk brane' $\widehat Q=A_G^c \otimes Q$ couples only to fields in the untwisted sector.} which then reduces to
\be\label{eq:chargeinduced}
\begin{tikzpicture}[very thick,scale=0.75,color=blue!50!black, baseline]
\nicepalecolourscheme (0,0) circle (1.5);
\draw (0,0) circle (1.5);
\fill (-45:1.55) circle (0pt) node[right] {{\small$\widehat Q$}};
\draw[<-, very thick] (0.100,-1.5) -- (-0.101,-1.5) node[above] {}; 
\draw[<-, very thick] (-0.100,1.5) -- (0.101,1.5) node[below] {}; 
\fill[color=green!50!black] (135:0) circle (2.5pt) node[right] {{\small$\alpha$}};
\draw[color=green!50!black] (0,0) .. controls +(0,0.6) and +(-0.4,-0.4) .. (45:1.5);
\fill[color=green!50!black] (45:1.5) circle (1.75pt) node[right] {{}};
\end{tikzpicture} 
\!\!\!=
\frac{|G|}{|H|} \;
\begin{tikzpicture}[very thick,scale=0.75,color=blue!50!black, baseline]
\nicepalecolourscheme (0,0) circle (1.5);
\draw (0,0) circle (1.5);
\fill (-45:1.55) circle (0pt) node[right] {{\small$Q$}};
\draw[<-, very thick] (0.100,-1.5) -- (-0.101,-1.5) node[above] {}; 
\draw[<-, very thick] (-0.100,1.5) -- (0.101,1.5) node[below] {}; 
\fill[color=green!50!black] (0,0) circle (2.5pt) node[left] {{\small$\alpha$}};
\draw[-dot-, color=green!50!black] (0,0) .. controls +(0,0.5) and +(0,0.5) .. (0.65,0);
\draw[color=green!50!black] (0.325,0.65) node[Odot] (down) {}; 
\draw[color=green!50!black] (down) -- (0.325,0.4); 
\draw[color=green!50!black] (1.3,0) .. controls +(0,0.1) and +(0,0) .. (30:1.5);
\fill[color=green!50!black] (30:1.5) circle (2.5pt) node[right] {{}};
\draw[-dot-, color=green!50!black] (0.65,0) .. controls +(0,-0.5) and +(0,-0.5) .. (1.3,0);
\draw[color=green!50!black] (0.975,-0.65) node[Odot] (down) {}; 
\draw[color=green!50!black] (down) -- (0.975,-0.4); 
\end{tikzpicture} 
\!\!\!=
\frac{|G|}{|H|} \;
\begin{tikzpicture}[very thick,scale=0.75,color=blue!50!black, baseline]
\nicepalecolourscheme (0,0) circle (1.5);
\draw (0,0) circle (1.5);
\fill (-45:1.55) circle (0pt) node[right] {{\small$Q$}};
\draw[<-, very thick] (0.100,-1.5) -- (-0.101,-1.5) node[above] {}; 
\draw[<-, very thick] (-0.100,1.5) -- (0.101,1.5) node[below] {}; 
\fill[color=green!50!black] (135:0) circle (2.5pt) node[left] {{}};
\fill[color=green!50!black] (0.1,0) circle (0pt) node[left] {{\small$\alpha$}};
\draw[color=green!50!black] (0,0) .. controls +(0,0.6) and +(-0.4,-0.4) .. (45:1.5);
\fill[color=green!50!black] (45:1.5) circle (2.5pt) node[right] {{}};
\end{tikzpicture} 
,
\ee
where we employed the Frobenius property for $A_H^c$ in the second step. In particular, this shows that finding the charges of the induced module $\widehat Q$ simply amounts to looking for fields $\alpha \in \Hom(I,A_H^c)$ with $\pi^{\text{RR}}_{A^c_G}(\alpha)=\alpha$ that couple to $Q$.

\medskip

We close this section by computing D-brane charges in a simple Landau-Ginzburg orbifold with discrete torsion. 
A discussion of the corresponding CFT is provided in Appendix~\ref{app:CFT}. 

\begin{example}\label{ex:inducedcharges}
We consider the potential $W=x_1^d+x_2^d$ with symmetry group $G'=\Z_d\times \Z_d$ acting on the variables by
$$
G'\ni (m,n): (x_1,x_2) \lmt (\omega^m x_1,\omega^n x_2) \, , \quad \omega=\E^{2 \pi \I/d} \, .
$$
Discrete torsion in this model is classified by $c \in \Z_d \cong H^2 (G',U(1))$. 
In this example we choose the group cohomology structure to be additive, so $c=0$ is the trivial class; then the discrete torsion phases are defined by 
\be\label{eq:discrphasesZdZd}
\varepsilon \big((m',n'),(m,n)\big)=\omega^{c(m'n-mn')} \, . 
\ee
One may also allow for nontrivial $\kappa \in \Hom(G',\Z_2)$, which determines the sign of the $G'$-action on untwisted RR ground states, as explained in Section~\ref{subsec:orbidefect}. In this case one has to assume $d$ to be even.

We restrict ourselves to induced $A_{G',\kappa}^c$-modules of the form $\widehat Q= A_{G',\kappa}^c \otimes_{A_G} Q$, where $Q$ is a matrix factorisation equivariant with respect to the diagonal subgroup $G\cong \Z_d$. For concreteness, we choose $Q$ to be the permutation matrix factorisation with 
$$
d_Q=\begin{pmatrix}0 & x-\mu y \\ \displaystyle \prod_{\mu'\neq \mu, \mu'^d=-1}(x-\mu' y) & 0\end{pmatrix}, \quad \gamma_{(m,m)}= \omega^{pm} \begin{pmatrix}1 & 0\\ 0 & \omega^m \end{pmatrix},
$$
where $\mu^d=-1$ and $p \in \Z_d$ denotes the representation label. 

Before orbifold projection, the spectrum of twisted RR ground states that couple to B-type D-branes in this model consists of the (unique) states $\phi_{(m,n)}$ in the $(m,n)$-twisted sectors with $m,n\neq 0$. The selection rule we have derived for induced modules further imposes $m=n$. Provided that $\phi_{(m,m)}$ is invariant in the $A_{G',\kappa}^c$-orbifold 
theory, we can use~\eqref{eq:chargeinduced} to compute the charge of $\widehat Q$ under this state as
$$
\big\langle \beta_{\widehat Q} (\phi_{(m,m)}) \big\rangle_{\widehat Q} =d \, \big\langle \beta_Q (\phi_{(m,m)}) \big\rangle_Q = -d \, \frac{\str_Q[ \gamma_{(m,m)}]}{(1-\omega)^2} \neq 0 \, ,
$$
where the second expression is evaluated directly using the residue formulae of~\cite{w0412274, BCP1}.\footnote{The factor $(1-\omega)^{-2}$ arises due to the normalisation of $\phi_{(m,m)}$, see \cite[App.\,A.3]{BCP1}.}

It remains to determine which of the states $\phi_{(m,m)}$ satisfy the invariance condition $(-1)^{\kappa_g(\kappa_h+1)}\varepsilon(g,h)=1$ for all $g \in G'$, where we denoted $h=(m,m)$. First, we consider the theory with $\kappa=0$. Choosing `minimal discrete torsion', that is $c=1$, we can read off from \eqref{eq:discrphasesZdZd} that none of the states survive the projection. On the other hand, for $c>1$, we obtain the requirement $\frac{d}{c}\in \Z$, hence there are $c-1$ invariant twisted states that have a nonvanishing coupling to the projective brane~$\widehat Q$. 

We also consider an example with $\kappa \neq 0$ (hence $d$ must be even), namely, we set $\kappa_{(1,0)}=\kappa_{(0,1)}=1$ and $c=1$. The invariance condition in this case has one solution, namely $m=\frac{d}{2}$, and by the argument above the corresponding RR ground state $\phi_{(\frac{d}{2},\frac{d}{2})}$ has again a nonzero coupling with $\widehat Q$.

Finally, let us note that the charge of $\widehat Q$ under $\phi_{(m,m)}$ can also be computed directly: 
$$
\big\langle \beta_{\widehat Q} (\phi_{(m,m)}) \big\rangle_{\widehat Q} = \big\langle \beta_Q (\phi_{(m,m)}) \big\rangle_Q \, \sum_{k \in \Z_d}   (-1)^{\kappa_{(0,k)}} \, \varepsilon \big((m,m),(0,k)\big)
$$
using the explicit form of the $c$-projective representation (analogous to~\eqref{eq:projaction2}, where here we take the group structure on $H^2 (G',U(1))$ to be multiplicative again)
$$
\widehat Q = \bigoplus_{k \in H} {}_k Q [\kappa_k] \, , 
\quad 
\widehat \gamma_{(h,g)} = \sum_{k \in H} \frac{c'(g,hk)}{c'(hk,g) \, c'(g,h)} \,  {}_{hk} (\gamma_g)[\kappa_{hk}] \circ {}_g(\pi_{k,hk})
$$ 
with $H=\{(0,l) \, | \, l\in \Z_d \} \subset G'$. 
\end{example}

\section{Orbifold equivalence in equivariant completion}\label{sec:orbifoldequivalences}

In this section we prove a general equivalence result in the setting of pivotal bicategories: 
any two objects $(a,A)$, $(a,A')$ in the equivariant completion of a pivotal bicategory are orbifold equivalent (under one mild assumption). 
For background and notation we refer to \cite[Sect.\,2\,\&\,4]{cr1210.6363}. 

\medskip

Let~$\B$ be a pivotal bicategory with idempotent complete categories of 1-morphisms. Recall that its equivariant completion $\Beq$ has objects which are pairs $(a,A)$ with $a\in \B$ and $A\in\B(a,a)$ a separable Frobenius algebra. 1- and 2-morphisms in $\Beq$ are bimodules and bimodule maps, respectively, horizontal composition is the tensor product over the intermediate Frobenius algebra, and the unit of $(a,A)$ is $I_{(a,A)} = A$. This is indeed a completion in the sense that there is an equivalence $\Beq \cong (\Beq)_{\mathrm{eq}}$. 

By assumption left and right adjoints of any $X\in \B(a,b)$ coincide, $X^\dagger = \dX$, but for a $B$-$A$-bimodule~$X$ its left and right adjoints as 1-morphisms in $\Beq$ are given by
$$
\deqX = {}_{\gamma_A^{-1}} (\dX) 
\, , \quad 
\ev_X = \!
\begin{tikzpicture}[very thick,scale=1.0,color=blue!50!black, baseline=.8cm]
\draw[line width=0pt] 
(1.75,1.75) node[line width=0pt, color=green!50!black] (A) {{\small$A\vphantom{\deqX }$}}
(1,0) node[line width=0pt] (D) {{\small$X\vphantom{\deqX }$}}
(0,0) node[line width=0pt] (s) {{\small$\deqX $}}; 
\draw[directed] (D) .. controls +(0,1.5) and +(0,1.5) .. (s);

\draw[color=green!50!black] (1.25,0.55) .. controls +(0.0,0.25) and +(0.25,-0.15) .. (0.86,0.95);
\draw[-dot-, color=green!50!black] (1.25,0.55) .. controls +(0,-0.5) and +(0,-0.5) .. (1.75,0.55);

\draw[color=green!50!black] (1.5,-0.1) node[Odot] (unit) {}; 
\draw[color=green!50!black] (1.5,0.15) -- (unit);

\fill[color=green!50!black] (0.86,0.95) circle (2pt) node (meet) {};

\draw[color=green!50!black] (1.75,0.55) -- (A);
\end{tikzpicture}
\circ \xi
\, , \quad
\coev_X =  \vartheta \circ 
\begin{tikzpicture}[very thick,scale=1.0,color=blue!50!black, baseline=-.8cm,rotate=180]
\draw[line width=0pt] 
(3.21,1.85) node[line width=0pt, color=green!50!black] (B) {{\small$B\vphantom{\deqX }$}}
(3,0) node[line width=0pt] (D) {{\small$X\vphantom{\deqX }$}}
(2,0) node[line width=0pt] (s) {{\small$\deqX $}}; 
\draw[redirected] (D) .. controls +(0,1.5) and +(0,1.5) .. (s);

\fill[color=green!50!black] (2.91,0.85) circle (2pt) node (meet) {};

\draw[color=green!50!black] (2.91,0.85) .. controls +(0.2,0.25) and +(0,-0.75) .. (B);

\end{tikzpicture}
$$
and similarly for $X^\star = (X^\dagger)_{\gamma_B}$. Here the maps $\xi: \deqX \otimes_B X \rightarrow \deqX  \otimes X$ and $\vartheta: X \otimes \deqX  \rightarrow X \otimes_A \deqX $ are splitting and projection 2-morphisms in~$\B$, and the notation $_{\alpha}(-)_{\beta}$ indicates that the left and right module actions are twisted by pre-composing with~$\alpha$ and~$\beta$, respectively. The relevant algebra maps for us are the Nakayama automorphism~$\gamma_A$ and its inverse: 
\be\label{eq:Nakayama}
\gamma_A = 
\begin{tikzpicture}[very thick, scale=0.5,color=green!50!black, baseline=-0.35cm]
\draw (0,0.8) -- (0,2);
\draw[-dot-] (0,0.8) .. controls +(0,-0.5) and +(0,-0.5) .. (-0.75,0.8);
\draw[directedgreen, color=green!50!black] (-0.75,0.8) .. controls +(0,0.5) and +(0,0.5) .. (-1.5,0.8);
\draw[-dot-] (0,-1.8) .. controls +(0,0.5) and +(0,0.5) .. (-0.75,-1.8);
\draw[redirectedgreen, color=green!50!black] (-0.75,-1.8) .. controls +(0,-0.5) and +(0,-0.5) .. (-1.5,-1.8);
\draw (0,-1.8) -- (0,-3);
\draw (-1.5,0.8) -- (-1.5,-1.8);
\draw (-0.375,-0.2) node[Odot] (D) {}; 
\draw (-0.375,0.4) -- (D);
\draw (-0.375,-0.8) node[Odot] (E) {}; 
\draw (-0.375,-1.4) -- (E);
\end{tikzpicture}
\, , \quad
\gamma_A^{-1} = 
\begin{tikzpicture}[very thick, scale=0.5,color=green!50!black, baseline=-0.35cm]
\draw (0,0.8) -- (0,2);
\draw[-dot-] (0,0.8) .. controls +(0,-0.5) and +(0,-0.5) .. (0.75,0.8);
\draw[directedgreen, color=green!50!black] (0.75,0.8) .. controls +(0,0.5) and +(0,0.5) .. (1.5,0.8);
\draw[-dot-] (0,-1.8) .. controls +(0,0.5) and +(0,0.5) .. (0.75,-1.8);
\draw[redirectedgreen, color=green!50!black] (0.75,-1.8) .. controls +(0,-0.5) and +(0,-0.5) .. (1.5,-1.8);
\draw (0,-1.8) -- (0,-3);
\draw (1.5,0.8) -- (1.5,-1.8);
\draw (0.375,-0.2) node[Odot] (D) {}; 
\draw (0.375,0.4) -- (D);
\draw (0.375,-0.8) node[Odot] (E) {}; 
\draw (0.375,-1.4) -- (E);
\end{tikzpicture}
\, . 
\ee

\medskip

Our aim is to show that under some condition any two objects $(a,A)$ and $(a,A')$ in $\Beq$ are orbifold equivalent, i.\,e.~there is an algebra~$\mathcal A$ and an isomorphism
\be\label{eq:Xschick}
\mathcal X : \big( (a,A), \mathcal A \big) \lra \big( (a,A'), I_{(a,A')} \big) \cong (a,A')
\quad \text{in } (\Beq)_{\mathrm{eq}} \cong \Beq. 
\ee
We will shortly describe~$\mathcal A$ and~$\mathcal X$ explicitly and give the precise statement in Theorem~\ref{thm:orbiequi}. One way to satisfy the condition mentioned above is to ask~$A'$ to be special, i.\,e.~$
\begin{tikzpicture}[very thick,scale=0.25,color=green!50!black, baseline]
\draw (0,-0.05) node[Odot] (D) {}; 
\draw (0,1.1) node[Odot] (E) {}; 
\draw (D) -- (E); 
\end{tikzpicture} 
\sim 1_{I_a}$. 

\medskip

We will work with the following variants of results obtained in \cite{cr1210.6363}. 

\begin{proposition}\label{prop:variant}
Let~$\B$ be a bicategory with adjoints, and $\mathcal X \in \B(a,b)$ such that there is an isomorphism $\alpha: \dsX \rightarrow \mathcal X^\dagger$.
\begin{enumerate}
\item The 1-morphism $\mathcal A = \mathcal X^\dagger \otimes \mathcal X$ is a Frobenius algebra. If furthermore $\tr_{\mathrm{r}}(\alpha) = 
\begin{tikzpicture}[very thick,scale=0.15,color=blue!50!black, baseline]
\draw (0,0) circle (2);
\draw[->, very thick] (0.100,-2) -- (-0.101,-2) node[above] {}; 
\draw[->, very thick] (0.100,2) -- (0.101,2) node[below] {}; 
\fill (0:2) circle (9.9pt) node[left] {{\small$\alpha$}};
\end{tikzpicture} 
$ 
is invertible, then~$\mathcal A$ is separable. 
\vspace{-0.2cm}
\item If $\tr_{\mathrm{r}}(\alpha)$ is invertible, then for $\mathcal A = \mathcal X^\dagger \otimes \mathcal X$ we have $\mathcal X \otimes_{\mathcal A} \mathcal X^\dagger \cong I_b$ and hence $(a,\mathcal A) \cong (b,I_b)$ in $\Beq$. 
\end{enumerate}
\end{proposition}

\begin{proof}
(i) 
We define the (co)multiplication and (co)unit on~$\mathcal A$ to be
$$
\begin{tikzpicture}[very thick,scale=0.9,color=blue!50!black, baseline=.9cm]
\draw[line width=0pt] 
(3,0) node[line width=0pt] (D) {{\small$\mathcal{X}^\dagger$}}
(2,0) node[line width=0pt] (s) {{\small$\mathcal{X}\vphantom{\mathcal{X}^\dagger}$}}; 
\draw[redirected] (D) .. controls +(0,1) and +(0,1) .. (s);
\draw[line width=0pt] 
(3.45,0) node[line width=0pt] (re) {{\small$\mathcal{X}\vphantom{\mathcal{X}^\dagger}$}}
(1.55,0) node[line width=0pt] (li) {{\small$\mathcal{X}^\dagger$}}; 
\draw[line width=0pt] 
(2.7,2) node[line width=0pt] (ore) {{\small$\mathcal{X}\vphantom{\mathcal{X}^\dagger}$}}
(2.3,2) node[line width=0pt] (oli) {{\small$\mathcal{X}^\dagger$}}; 
\draw (li) .. controls +(0,0.75) and +(0,-0.25) .. (2.3,1.25);
\draw (2.3,1.25) -- (oli);
\draw (re) .. controls +(0,0.75) and +(0,-0.25) .. (2.7,1.25);
\draw (2.7,1.25) -- (ore);
\end{tikzpicture}
, \quad
\begin{tikzpicture}[very thick,scale=1.0,color=blue!50!black, baseline=-.4cm,rotate=180]
\draw[line width=0pt] 
(3,0) node[line width=0pt] (D) {{\small$\mathcal{X}^\dagger$}}
(2,0) node[line width=0pt] (s) {{\small$\mathcal{X}\vphantom{\mathcal{X}^\dagger}$}}; 
\draw[directed] (D) .. controls +(0,1) and +(0,1) .. (s);
\end{tikzpicture}
, \quad
\begin{tikzpicture}[very thick,scale=0.9,color=blue!50!black, baseline=-0.9cm, rotate=180]
\draw[line width=0pt] 
(3,0) node[line width=0pt] (D) {{\small$\mathcal{X}\vphantom{\mathcal{X}^\dagger}$}}
(2,0) node[line width=0pt] (s) {{\small${\mathcal{X}^\dagger}$}}; 
\draw[redirected] (D) .. controls +(0,1) and +(0,1) .. (s);
\draw (2.5,1.13) node (D) {{\small$\boldsymbol{\star}$}}; 
\draw[line width=0pt] 
(3.45,0) node[line width=0pt] (re) {{\small${\mathcal{X}^\dagger}$}}
(1.55,0) node[line width=0pt] (li) {{\small$\mathcal{X}\vphantom{\mathcal{X}^\dagger}$}}; 
\draw[line width=0pt] 
(2.7,2) node[line width=0pt] (ore) {{\small${\mathcal{X}^\dagger}$}}
(2.3,2) node[line width=0pt] (oli) {{\small$\mathcal{X}\vphantom{\mathcal{X}^\dagger}$}}; 
\draw (li) .. controls +(0,0.75) and +(0,-0.25) .. (2.3,1.25);
\draw (2.3,1.25) -- (oli);
\draw (re) .. controls +(0,0.75) and +(0,-0.25) .. (2.7,1.25);
\draw (2.7,1.25) -- (ore);
\fill (2.03,0.5) circle (2.2pt) node[left] {{\small$\alpha$}};
\end{tikzpicture}
\, , \quad 
\begin{tikzpicture}[very thick,scale=1.0,color=blue!50!black, baseline=.4cm]
\draw[line width=0pt] 
(3,0) node[line width=0pt] (D) {{\small$\mathcal{X}\vphantom{\mathcal{X}^\dagger}$}}
(2,0) node[line width=0pt] (s) {{\small$\mathcal X^\dagger$}}; 
\draw[directed] (D) .. controls +(0,1) and +(0,1) .. (s);
\fill (2.03,0.5) circle (2.2pt) node[right] {{\small$\alpha^{-1}$}};
\end{tikzpicture}
\circ (1_{\mathcal{X}^\dagger} \otimes \tr_{\mathrm{r}}(\alpha) \otimes 1_\mathcal{X})
$$
with $\boldsymbol{\star} = \tr_{\mathrm{r}}(\alpha)^{-1}$. Then the proof works analogously to the one of \cite[Thm.\,4.3]{cr1210.6363}. 

(ii) Same as for \cite[Thm.\,4.4]{cr1210.6363} after (two) appropriate insertions of~$\alpha$. 
\end{proof}

In a pivotal bicategory~$\B$ we now pick two separable Frobenius algebras $A,A'\in\B(a,a)$ and define 
$$
\mathcal X = A' \otimes A \, . 
$$
Endowed with the obvious $A'$-$A$-bimodule structure~$\mathcal X$ is a 1-morphism in $\Beq$, and we have $\mathcal X^\star = (A' \otimes A)^\dagger_{\gamma_{A'}} \cong {}_{\gamma_A^{-1}} {}^{\dagger} (A' \otimes A) = \dseqX$: 

\begin{lemma}
\be\label{eq:alpha}
\alpha = 
\begin{tikzpicture}[very thick,scale=0.65,color=blue!50!black, baseline=-0.925cm]
\draw[line width=1pt] 
(-1,2) node[line width=0pt] (XY) {{\small $\mathcal X^\star$}}; 
\draw[directedgreen] (1,0) .. controls +(0,1) and +(0,1) .. (2,0);
\draw[directedgreen] (0,0) .. controls +(0,2) and +(0,2) .. (3,0);
\draw[directed, ultra thick] (-1,0) .. controls +(0,-1) and +(0,-1) .. (0.5,0);
\draw[color=green!50!black] (2,0) -- (2,-2.5);
\draw[color=green!50!black] (3,0) -- (3,-2.5);
\draw[dotted] (0,0) -- (1,0);
\draw[ultra thick] (-1,0) -- (XY);
\draw[redirectedgreen] (1,-2.5) .. controls +(0,-1) and +(0,-1) .. (2,-2.5);
\draw[redirectedgreen] (0,-2.5) .. controls +(0,-2) and +(0,-2) .. (3,-2.5);
\draw[redirected, ultra thick] (-1,-2.5) .. controls +(0,1) and +(0,1) .. (0.5,-2.5);
\draw[dotted] (0,-2.5) -- (1,-2.5);
\draw[line width=1pt] 
(-1,-4.5) node[line width=0pt] (XYdown) {{\small $\dseqX$}}; 
\draw[ultra thick] (-1,-2.5) -- (XYdown);
\fill[color=green!50!black] (2,-1.25) circle (2.5pt) node[left] {{\small $\gamma_{\dA}$}};
\fill[color=green!50!black] (3,-1.25) circle (2.5pt) node[right] {{\small $\gamma_{\dA'}$}};
\end{tikzpicture}
\ee
is an isomorphism in $\Beq$. 
\end{lemma}

\begin{proof}
We have to show that 
$$
\gamma_{\dA} = 
\begin{tikzpicture}[very thick, scale=0.5,color=green!50!black, baseline=-0.35cm]
\draw (0,0.8) -- (0,2);
\draw[directedgreen] (0,0.8) .. controls +(0,-0.5) and +(0,-0.5) .. (0.75,0.8);
\draw[-dot-, color=green!50!black] (0.75,0.8) .. controls +(0,0.5) and +(0,0.5) .. (1.5,0.8);
\draw[redirectedgreen] (0,-1.8) .. controls +(0,0.5) and +(0,0.5) .. (0.75,-1.8);
\draw[-dot-, color=green!50!black] (0.75,-1.8) .. controls +(0,-0.5) and +(0,-0.5) .. (1.5,-1.8);
\draw (0,-1.8) -- (0,-3);
\draw (1.5,0.8) -- (1.5,-1.8);
\draw (1.125,1.78) node[Odot] (D) {}; 
\draw (1.125,1.25) -- (D);
\draw (1.125,-2.78) node[Odot] (E) {}; 
\draw (1.125,-2.1) -- (E);
\end{tikzpicture}
: {}_{\gamma_A^{-1}} \dA \lra A^\dagger 
$$
is a map of left $A$-modules, and that $\gamma_{A'} : \dA' \rightarrow A'^\dagger_{\gamma_{A'}}$ is a map of right $A'$-modules. We will work out the details only for the former case, and use $A^\dagger = \dA$ from now on. 

We start by writing the left $A$-action on ${}_{\gamma_A^{-1}} A^\dagger$ in a different form: 
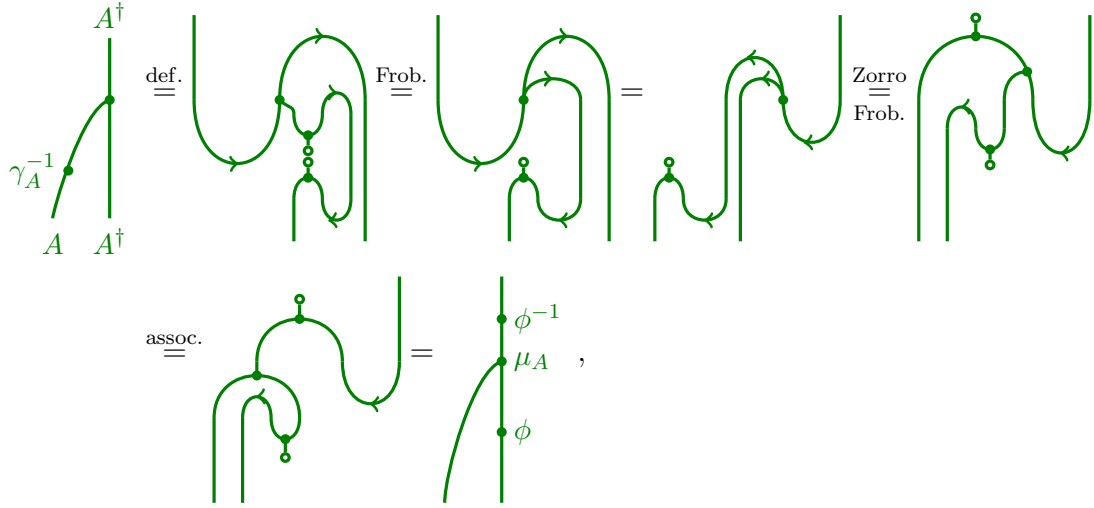
\begin{align*}
\begin{tikzpicture}[very thick,scale=0.75,color=green!50!black, baseline=0cm]
\draw[line width=0] 
(-1,-2.5) node[line width=0pt] (A) {{\small $A\vphantom{A^\dagger}$}}
(0,-2.5) node[line width=0pt] (Ad) {{\small $A^\dagger$}}
(0,1.5) node[line width=0pt] (Au) {{\small $A^\dagger$}};
\draw (Ad) -- (Au);
\fill (0,0) circle (2.5pt) node (meet) {};
\draw[color=green!50!black]  (0,0) .. controls +(0.-0.5,-0.15) and +(0,0.5) .. (A);
\fill (-0.725,-1.25) circle (2.5pt) node[left] (meet) {{\small $\gamma_A^{-1}$}};
\end{tikzpicture}
& 
\stackrel{\mathrm{def.}}{=}
\begin{tikzpicture}[very thick,scale=0.75,color=green!50!black, baseline=0cm]
\draw[directedgreen, color=green!50!black] (0,0) .. controls +(0,1.5) and +(0,1.5) .. (1.5,0);
\draw[directedgreen, color=green!50!black] (-1.5,0) .. controls +(0,-1.5) and +(0,-1.5) .. (0,0);
\fill[color=green!50!black]  (0,0) circle (2.5pt) node (meet) {};
\draw[color=green!50!black]  (0,0) .. controls +(0.15,-0.15) and +(0,0.15) .. (0.25,-0.25);
\draw[-dot-] (0.25,-0.25) .. controls +(0,-0.5) and +(0,-0.5) .. (0.75,-0.25);
\draw (0.5,-0.9) node[Odot] (D) {}; 
\draw (0.5,-0.7) -- (D);
\draw[directedgreen, color=green!50!black] (0.75,-0.25) .. controls +(0,0.5) and +(0,0.5) .. (1.25,-0.25);
\draw[-dot-] (0.25,-1.75) .. controls +(0,0.5) and +(0,0.5) .. (0.75,-1.75);
\draw (0.5,-1.1) node[Odot] (E) {}; 
\draw (0.5,-1.4) -- (E);
\draw[redirectedgreen, color=green!50!black] (0.75,-1.75) .. controls +(0,-0.5) and +(0,-0.5) .. (1.25,-1.75);
\draw (1.25,-0.25) -- (1.25,-1.75);
\draw (-1.5,0) -- (-1.5,1.5);
\draw (1.5,0) -- (1.5,-2.5);
\draw (0.25,-1.75) -- (0.25,-2.5);
\end{tikzpicture}
\stackrel{\mathrm{Frob.}}{=}
\begin{tikzpicture}[very thick,scale=0.75,color=green!50!black, baseline=0cm]
\draw[directedgreen, color=green!50!black] (0,0) .. controls +(0,1.5) and +(0,1.5) .. (1.5,0);
\draw[directedgreen, color=green!50!black] (-1.5,0) .. controls +(0,-1.5) and +(0,-1.5) .. (0,0);
\draw (-1.5,0) -- (-1.5,1.5);
\draw (1.5,0) -- (1.5,-2.5);
\fill[color=green!50!black]  (0,0) circle (2.5pt) node (meet) {};
\draw[directedgreen, color=green!50!black] (0,0) .. controls +(0,0.5) and +(0,0.5) .. (1,0);
\draw (1,0) -- (1,-1.75);
\draw[directedgreen, color=green!50!black] (1,-1.75) .. controls +(0,-0.5) and +(0,-0.5) .. (0.25,-1.75);
\draw[-dot-] (0.25,-1.75) .. controls +(0,0.5) and +(0,0.5) .. (-0.25,-1.75);
\draw (0,-1.1) node[Odot] (E) {}; 
\draw (0,-1.4) -- (E);
\draw (-0.25,-1.75) -- (-0.25,-2.5);
\end{tikzpicture}
=
\begin{tikzpicture}[very thick,scale=0.75,color=green!50!black, baseline=0cm]
\draw (1,0) -- (1,-1.75);
\draw[directedgreen, color=green!50!black] (1,-1.75) .. controls +(0,-0.5) and +(0,-0.5) .. (0.25,-1.75);
\draw[-dot-] (0.25,-1.75) .. controls +(0,0.5) and +(0,0.5) .. (-0.25,-1.75);
\draw (0,-1.1) node[Odot] (E) {}; 
\draw (0,-1.4) -- (E);
\draw (-0.25,-1.75) -- (-0.25,-2.5);
\draw[redirectedgreen, color=green!50!black] (1,0) .. controls +(0,1) and +(0,1) .. (2,0);
\draw[redirectedgreen, color=green!50!black] (2,0) .. controls +(0,-1) and +(0,-1) .. (3,0);
\draw (3,0) -- (3,1.5);
\fill[color=green!50!black]  (2,0) circle (2.5pt) node (meet) {};
\draw[directedgreen, color=green!50!black] (2,0) .. controls +(0,0.5) and +(0,0.5) .. (1.25,0);
\draw (1.25,0) -- (1.25,-2.5);
\end{tikzpicture}
\,
\eq^{\mathrm{Zorro}}_{\mathrm{Frob.}}
\,
\begin{tikzpicture}[very thick,scale=0.75,color=green!50!black, baseline=0cm]
\draw (1,0) -- (1,1.5);
\draw[-dot-, color=green!50!black] (0,0) .. controls +(0,1.5) and +(0,1.5) .. (-2,0);
\draw[directedgreen, color=green!50!black] (1,0) .. controls +(0,-1.25) and +(0,-1.25) .. (0,0);
\draw (-1,1.45) node[Odot] (E) {}; 
\draw (-1,1.15) -- (E);
\fill[color=green!50!black]  (-0.1,0.5) circle (2.5pt) node (meet) {};
\draw[color=green!50!black] (-0.1,0.5) .. controls +(-0.25,0) and +(0,0.3) .. (-0.5,0);
\draw (-0.5,0) -- (-0.5,-0.5);
\draw[-dot-] (-0.5,-0.5) .. controls +(0,-0.5) and +(0,-0.5) .. (-1,-0.5);
\draw (-0.75,-1.15) node[Odot] (D) {}; 
\draw (-0.75,-0.9) -- (D);
\draw[directedgreen, color=green!50!black] (-1,-0.5) .. controls +(0,0.5) and +(0,0.5) .. (-1.5,-0.5);
\draw (-1.5,-0.5) -- (-1.5,-2.5);
\draw (-2,0) -- (-2,-2.5);
\end{tikzpicture}
\\
& 
\stackrel{\mathrm{assoc.}}{=}
\begin{tikzpicture}[very thick,scale=0.75,color=green!50!black, baseline=0cm]
\draw (1,0) -- (1,1.5);
\draw[directedgreen] (1,0) .. controls +(0,-1) and +(0,-1) .. (0,0);
\draw[-dot-] (0,0) .. controls +(0,1) and +(0,1) .. (-1.5,0);
\draw (-0.75,1.1) node[Odot] (E) {}; 
\draw (-0.75,0.8) -- (E);
\draw[-dot-] (-0.75,-1) .. controls +(0,1) and +(0,1) .. (-2.25,-1);
\draw[-dot-] (-0.75,-1) .. controls +(0,-0.5) and +(0,-0.5) .. (-1.25,-1);
\draw (-1,-1.7) node[Odot] (F) {}; 
\draw (-1,-1.4) -- (F);
\draw[directedgreen] (-1.25,-1) .. controls +(0,0.5) and +(0,0.5) .. (-1.75,-1);
\draw (-1.5,0) -- (-1.5,-0.25);
\draw (-1.75,-1) -- (-1.75,-2.5);
\draw (-2.25,-1) -- (-2.25,-2.5);
\end{tikzpicture}
=
\begin{tikzpicture}[very thick,scale=0.75,color=green!50!black, baseline=0cm]
\draw (0,1.5) -- (0,-2.5);
\fill (0,0) circle (2.5pt) node[right] (meet) {{\small $\mu_A$}};
\draw[color=green!50!black]  (0,0) .. controls +(0.-0.5,-0.15) and +(0,0.5) .. (-1,-2.5);
\fill (0,-1.25) circle (2.5pt) node[right] (meet) {{\small $\phi$}};
\fill (0,0.75) circle (2.5pt) node[right] (meet) {{\small $\phi^{-1}$}};
\end{tikzpicture}
\, , 
\end{align*}
where in the last step we wrote~$\mu_A$ for the multiplication on~$A$, and we introduced the algebra maps 
$$
\phi = 
\begin{tikzpicture}[very thick,scale=0.5,color=green!50!black, baseline=0cm]
\draw[directedgreen, color=green!50!black] (0,0) .. controls +(0,1) and +(0,1) .. (-1,0);
\draw[-dot-] (1,0) .. controls +(0,-1) and +(0,-1) .. (0,0);
\draw (-1,0) -- (-1,-1.3); 
\draw (1,0) -- (1,1.3); 
\draw (0.5,-1.2) node[Odot] (end) {}; 
\draw (0.5,-0.8) -- (end); 
\end{tikzpicture}
: A^\dagger \lra A \, , 
\quad 
\phi^{-1} = 
\begin{tikzpicture}[very thick,scale=0.5,color=green!50!black, baseline=0cm]
\draw[-dot-] (0,0) .. controls +(0,1) and +(0,1) .. (-1,0);
\draw[directedgreen, color=green!50!black] (1,0) .. controls +(0,-1) and +(0,-1) .. (0,0);
\draw (-1,0) -- (-1,-1.3); 
\draw (1,0) -- (1,1.3); 
\draw (-0.5,1.2) node[Odot] (end) {}; 
\draw (-0.5,0.8) -- (end); 
\end{tikzpicture}
: A \lra A^\dagger \, . 
$$
We thus find that
\be\label{eq:leftactiontwisted}
\begin{tikzpicture}[very thick,scale=0.75,color=green!50!black, baseline=0cm]
\draw[line width=0] 
(-1,-1.5) node[line width=0pt] (A) {{\small $A\vphantom{A^\dagger}$}}
(0,-1.5) node[line width=0pt] (Ad) {{\small $A^\dagger$}}
(0,1.5) node[line width=0pt] (Au) {{\small $A^\dagger$}};
\draw (Ad) -- (Au);
\fill (0,0) circle (2.5pt) node (meet) {};
\draw[color=green!50!black]  (0,0) .. controls +(0.-0.5,-0.15) and +(0,0.5) .. (A);
\fill (-0.725,-0.63) circle (2.5pt) node[left] (meet) {{\small $\gamma_A^{-1}$}};
\end{tikzpicture}
=
\begin{tikzpicture}[very thick,scale=0.75,color=green!50!black, baseline=0cm]
\draw[line width=0] 
(-1,-1.5) node[line width=0pt] (A) {{\small $A\vphantom{A^\dagger}$}}
(0,-1.5) node[line width=0pt] (Ad) {{\small $A^\dagger$}}
(0,1.5) node[line width=0pt] (Au) {{\small $A^\dagger$}};
\draw (Ad) -- (Au);
\fill (0,0) circle (2.5pt) node[right] (meet) {{\small $\mu_{A^\dagger}$}};
\draw[color=green!50!black]  (0,0) .. controls +(0.-0.5,-0.15) and +(0,0.5) .. (A);
\fill (-0.725,-0.63) circle (2.5pt) node[left] (meet) {{\small $\phi^{-1}$}};
\end{tikzpicture}
\, . 
\ee

Now we can show that~$\gamma_{A^\dagger}$ is a module map: 
\begin{align*}
\begin{tikzpicture}[very thick,scale=0.75,color=green!50!black, baseline=0cm]
\draw[line width=0] 
(-1,-1.5) node[line width=0pt] (A) {{\small $A\vphantom{{}_{\gamma_A^{-1}} A^\dagger}$}}
(0,-1.5) node[line width=0pt] (Ad) {{\small ${}_{\gamma_A^{-1}} A^\dagger$}}
(0,1.5) node[line width=0pt] (Au) {{\small $A^\dagger$}};
\draw (Ad) -- (Au);
\fill (0,0) circle (2.5pt) node (meet) {};
\draw[color=green!50!black]  (0,0) .. controls +(0.-0.5,-0.15) and +(0,0.5) .. (A);
\fill (0,0.75) circle (2.5pt) node[right] (meet) {{\small $\gamma_{A^\dagger}$}};
\end{tikzpicture}
=
\begin{tikzpicture}[very thick,scale=0.75,color=green!50!black, baseline=0cm]
\draw[line width=0] 
(-1,-1.5) node[line width=0pt] (A) {{\small $A\vphantom{A^\dagger}$}}
(0,-1.5) node[line width=0pt] (Ad) {{\small $A^\dagger$}}
(0,1.5) node[line width=0pt] (Au) {{\small $A^\dagger$}};
\draw (Ad) -- (Au);
\fill (0,0) circle (2.5pt) node (meet) {};
\draw[color=green!50!black]  (0,0) .. controls +(0.-0.5,-0.15) and +(0,0.5) .. (A);
\fill (-0.725,-0.63) circle (2.5pt) node[left] (meet) {{\small $\gamma_A^{-1}$}};
\fill (0,0.75) circle (2.5pt) node[right] (meet) {{\small $\gamma_{A^\dagger}$}};
\end{tikzpicture}
\eq^{\eqref{eq:leftactiontwisted}}
\begin{tikzpicture}[very thick,scale=0.75,color=green!50!black, baseline=0cm]
\draw[line width=0] 
(-1,-1.5) node[line width=0pt] (A) {{\small $A\vphantom{A^\dagger}$}}
(0,-1.5) node[line width=0pt] (Ad) {{\small $A^\dagger$}}
(0,1.5) node[line width=0pt] (Au) {{\small $A^\dagger$}};
\draw (Ad) -- (Au);
\fill (0,0) circle (2.5pt) node[right] (meet) {{\small $\mu_{A^\dagger}$}};
\draw[color=green!50!black]  (0,0) .. controls +(0.-0.5,-0.15) and +(0,0.5) .. (A);
\fill (-0.725,-0.63) circle (2.5pt) node[left] (meet) {{\small $\phi^{-1}$}};
\fill (0,0.75) circle (2.5pt) node[right] (meet) {{\small $\gamma_{A^\dagger}$}};
\end{tikzpicture}
=
\begin{tikzpicture}[very thick,scale=0.75,color=green!50!black, baseline=0cm]
\draw[line width=0] 
(-1,-1.5) node[line width=0pt] (A) {{\small $A\vphantom{A^\dagger}$}}
(0,-1.5) node[line width=0pt] (Ad) {{\small $A^\dagger$}}
(0,1.5) node[line width=0pt] (Au) {{\small $A^\dagger$}};
\draw (Ad) -- (Au);
\fill (0,0) circle (2.5pt) node[right] (meet) {{\small $\mu_{A^\dagger}$}};
\draw[color=green!50!black]  (0,0) .. controls +(0.-0.5,-0.15) and +(0,0.5) .. (A);
\fill (-0.81,-0.73) circle (2.5pt) node[left] (meet) {{\small $\phi^{-1}$}};
\fill (0,-0.75) circle (2.5pt) node[right] (meet) {{\small $\gamma_{A^\dagger}$}};
\fill (-0.35,-0.19) circle (2.5pt) node[left] (meet) {{\small $\gamma_{A^\dagger}$}};
\end{tikzpicture}
\eq^{\eqref{eq:leftactiontwisted}}
\begin{tikzpicture}[very thick,scale=0.75,color=green!50!black, baseline=0cm]
\draw[line width=0] 
(-1,-1.5) node[line width=0pt] (A) {{\small $A\vphantom{A^\dagger}$}}
(0,-1.5) node[line width=0pt] (Ad) {{\small $A^\dagger$}}
(0,1.5) node[line width=0pt] (Au) {{\small $A^\dagger$}};
\draw (Ad) -- (Au);
\draw[color=green!50!black]  (0,0) .. controls +(0.-0.5,-0.15) and +(0,0.5) .. (A);
\fill (0,-0.75) circle (2.5pt) node[right] (meet) {{\small $\gamma_{A^\dagger}$}};
\end{tikzpicture} 
.
\end{align*}
In the last step we first used $\gamma_{A^\dagger} \circ \phi^{-1} = \phi^{-1} \circ \gamma_A$, which in turn follows straightforwardly from the definitions. 
\end{proof}

\begin{theorem}\label{thm:orbiequi}
Let $(a,A)$, $(a,A') \in \Beq$, set $\mathcal X = A' \otimes A$ and let $\alpha: \dseqX \rightarrow \mathcal X^\star$ as in~\eqref{eq:alpha} have invertible $\tr_{\mathrm{r}}(\alpha)$. Then with $\mathcal A = \mathcal X^\star \otimes_{A'} \mathcal X$ we have 
$$
\big( (a,A), \mathcal A \big) \cong \big( (a,A'), A' \big) 
\quad \text{in } (\Beq)_{\mathrm{eq}} \cong \Beq . 
$$
In particular $\tr_{\mathrm{r}}(\alpha)$ is invertible if~$A'$ is special, i.\,e.~$
\begin{tikzpicture}[very thick,scale=0.25,color=green!50!black, baseline]
\draw (0,-0.05) node[Odot] (D) {}; 
\draw (0,1.1) node[Odot] (E) {}; 
\draw (D) -- (E); 
\end{tikzpicture} 
$ is a nonzero multiple of~$1_{I_a}$. 
\end{theorem}

\begin{proof}
The general statement immediately follows from Proposition~\ref{prop:variant}(ii) applied to the bicategory $\Beq$. 

For the last part we compute $\tr_{\mathrm{r}}(\alpha)$ to be
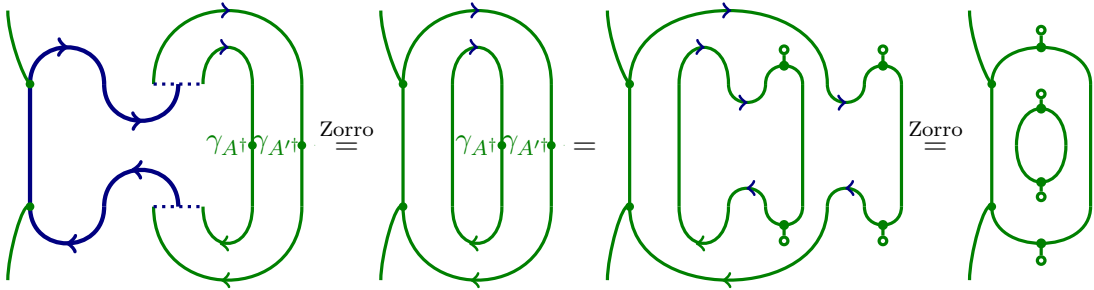
\begin{align*}
\begin{tikzpicture}[very thick,scale=0.65,color=blue!50!black, baseline=-0.925cm]
\draw[directedgreen] (1,0) .. controls +(0,1) and +(0,1) .. (2,0);
\draw[directedgreen] (0,0) .. controls +(0,2) and +(0,2) .. (3,0);
\draw[directed, ultra thick] (-1,0) .. controls +(0,-1) and +(0,-1) .. (0.5,0);
\draw[color=green!50!black] (2,0) -- (2,-2.5);
\draw[color=green!50!black] (3,0) -- (3,-2.5);
\draw[dotted] (0,0) -- (1,0);
\draw[redirectedgreen] (1,-2.5) .. controls +(0,-1) and +(0,-1) .. (2,-2.5);
\draw[redirectedgreen] (0,-2.5) .. controls +(0,-2) and +(0,-2) .. (3,-2.5);
\draw[redirected, ultra thick] (-1,-2.5) .. controls +(0,1) and +(0,1) .. (0.5,-2.5);
\draw[dotted] (0,-2.5) -- (1,-2.5);
\fill[color=green!50!black] (2,-1.25) circle (2.5pt) node[left] {};
\fill[color=green!50!black] (2.2,-1.25) circle (0pt) node[left] {{\small $\gamma_{A^\dagger}$}};
\fill[color=green!50!black] (3,-1.25) circle (2.5pt) node[left] {};
\fill[color=green!50!black] (3.25,-1.25) circle (0pt) node[left] {{\small $\gamma_{A'^\dagger}$}};
\draw[directed, ultra thick] (-2.5,0) .. controls +(0,1) and +(0,1) .. (-1,0);
\draw[directed, ultra thick] (-1,-2.5) .. controls +(0,-1) and +(0,-1) .. (-2.5,-2.5);
\draw[ultra thick] (-2.5,-2.5) -- (-2.5,0);
\fill[color=green!50!black] (-2.5,-2.5) circle (2.5pt) node[left] {};
\fill[color=green!50!black] (-2.5,0) circle (2.5pt) node[left] {};
\draw[color=green!50!black]  (-2.95,-4) .. controls +(0,0.5) and +(-0.15,0.15) .. (-2.5,-2.5);
\draw[color=green!50!black]  (-2.95,1.5) .. controls +(0,-0.5) and +(-0.15,0.15) .. (-2.5,0);
\end{tikzpicture}
\eq^{\mathrm{Zorro}}
\begin{tikzpicture}[very thick,scale=0.65,color=blue!50!black, baseline=-0.925cm]
\draw[directedgreen] (1,0) .. controls +(0,1) and +(0,1) .. (2,0);
\draw[directedgreen] (0,0) .. controls +(0,2) and +(0,2) .. (3,0);
\draw[color=green!50!black] (2,0) -- (2,-2.5);
\draw[color=green!50!black] (3,0) -- (3,-2.5);
\draw[color=green!50!black] (1,0) -- (1,-2.5);
\draw[color=green!50!black] (0,0) -- (0,-2.5);
\draw[redirectedgreen] (1,-2.5) .. controls +(0,-1) and +(0,-1) .. (2,-2.5);
\draw[redirectedgreen] (0,-2.5) .. controls +(0,-2) and +(0,-2) .. (3,-2.5);
\fill[color=green!50!black] (2,-1.25) circle (2.5pt) node[left] {};
\fill[color=green!50!black] (2.2,-1.25) circle (0pt) node[left] {{\small $\gamma_{A^\dagger}$}};
\fill[color=green!50!black] (3,-1.25) circle (2.5pt) node[left] {};
\fill[color=green!50!black] (3.25,-1.25) circle (0pt) node[left] {{\small $\gamma_{A'^\dagger}$}};
\fill[color=green!50!black] (-0,-2.5) circle (2.5pt) node[left] {};
\fill[color=green!50!black] (-0,0) circle (2.5pt) node[left] {};
\draw[color=green!50!black]  (-0.45,-4) .. controls +(0,0.5) and +(-0.15,0.15) .. (0,-2.5);
\draw[color=green!50!black]  (-0.45,1.5) .. controls +(0,-0.5) and +(-0.15,0.15) .. (0,0);
\end{tikzpicture}
=
\begin{tikzpicture}[very thick,scale=0.65,color=blue!50!black, baseline=-0.925cm]
\draw[directedgreen] (1,0) .. controls +(0,1) and +(0,1) .. (2,0);
\draw[directedgreen] (0,0) .. controls +(0,2) and +(0,2) .. (4,0);
%
%
\draw[directedgreen] (2,0) .. controls +(0,-0.5) and +(0,-0.5) .. (2.75,0);
\draw[-dot-, color=green!50!black] (2.75,0) .. controls +(0,0.5) and +(0,0.5) .. (3.5,0);
\draw[color=green!50!black] (3.125,0.7) node[Odot] (D) {}; 
\draw[color=green!50!black] (3.125,0.4) -- (D);
\draw[color=green!50!black] (3.5,-2.5) -- (3.5,0);
\draw[-dot-, color=green!50!black] (2.75,-2.5) .. controls +(0,-0.5) and +(0,-0.5) .. (3.5,-2.5);
\draw[color=green!50!black] (3.125,-3.2) node[Odot] (D) {}; 
\draw[color=green!50!black] (3.125,-2.9) -- (D);
\draw[directedgreen] (2.75,-2.5) .. controls +(0,0.5) and +(0,0.5) .. (2,-2.5);
\draw[color=green!50!black] (1,0) -- (1,-2.5);
\draw[color=green!50!black] (0,0) -- (0,-2.5);
\draw[redirectedgreen] (1,-2.5) .. controls +(0,-1) and +(0,-1) .. (2,-2.5);
\draw[redirectedgreen] (0,-2.5) .. controls +(0,-2) and +(0,-2) .. (4,-2.5);
\draw[directedgreen] (4,0) .. controls +(0,-0.5) and +(0,-0.5) .. (4.75,0);
\draw[-dot-, color=green!50!black] (4.75,0) .. controls +(0,0.5) and +(0,0.5) .. (5.5,0);
\draw[color=green!50!black] (5.125,0.7) node[Odot] (D) {}; 
\draw[color=green!50!black] (5.125,0.4) -- (D);
\draw[color=green!50!black] (5.5,-2.5) -- (5.5,0);
\draw[-dot-, color=green!50!black] (4.75,-2.5) .. controls +(0,-0.5) and +(0,-0.5) .. (5.5,-2.5);
\draw[color=green!50!black] (5.125,-3.2) node[Odot] (D) {}; 
\draw[color=green!50!black] (5.125,-2.9) -- (D);
\draw[directedgreen] (4.75,-2.5) .. controls +(0,0.5) and +(0,0.5) .. (4,-2.5);
\fill[color=green!50!black] (-0,-2.5) circle (2.5pt) node[left] {};
\fill[color=green!50!black] (-0,0) circle (2.5pt) node[left] {};
\draw[color=green!50!black]  (-0.45,-4) .. controls +(0,0.5) and +(-0.15,0.15) .. (0,-2.5);
\draw[color=green!50!black]  (-0.45,1.5) .. controls +(0,-0.5) and +(-0.15,0.15) .. (0,0);
\end{tikzpicture}
\eq^{\mathrm{Zorro}}
\begin{tikzpicture}[very thick,scale=0.65,color=blue!50!black, baseline=-0.925cm]
\fill[color=green!50!black] (-0,-2.5) circle (2.5pt) node[left] {};
\fill[color=green!50!black] (-0,0) circle (2.5pt) node[left] {};
\draw[color=green!50!black]  (-0.45,-4) .. controls +(0,0.5) and +(-0.15,0.15) .. (0,-2.5);
\draw[color=green!50!black]  (-0.45,1.5) .. controls +(0,-0.5) and +(-0.15,0.15) .. (0,0);
\draw[color=green!50!black] (0,0) -- (0,-2.5);
\draw[-dot-, color=green!50!black] (0,0) .. controls +(0,1) and +(0,1) .. (2,0);
\draw[color=green!50!black] (1,1.1) node[Odot] (D) {}; 
\draw[color=green!50!black] (1,0.8) -- (D);
\draw[-dot-, color=green!50!black] (0,-2.5) .. controls +(0,-1) and +(0,-1) .. (2,-2.5);
\draw[color=green!50!black] (1,-3.6) node[Odot] (D) {}; 
\draw[color=green!50!black] (1,-3.3) -- (D);
\draw[color=green!50!black] (2,0) -- (2,-2.5);
\draw[-dot-, color=green!50!black] (0.5,-1.25) .. controls +(0,1) and +(0,1) .. (1.5,-1.25);
\draw[color=green!50!black] (1,-0.2) node[Odot] (D) {}; 
\draw[color=green!50!black] (1,-0.5) -- (D);
\draw[-dot-, color=green!50!black] (0.5,-1.25) .. controls +(0,-1) and +(0,-1) .. (1.5,-1.25);
\draw[color=green!50!black] (1,-2.3) node[Odot] (D) {}; 
\draw[color=green!50!black] (1,-2.1) -- (D);
\end{tikzpicture}
\end{align*}
which by separability and the Frobenius property reduces to~$1_A$ multiplied by the number $
\begin{tikzpicture}[very thick,scale=0.25,color=green!50!black, baseline]
\draw (0,-0.05) node[Odot] (D) {}; 
\draw (0,1.1) node[Odot] (E) {}; 
\draw (D) -- (E); 
\end{tikzpicture} 
$ 
if~$A'$ is special. 
\end{proof}

Since every Frobenius algebra is self-dual we have $\mathcal A \cong A \otimes A' \otimes A$. In $\Beq$ not every two objects are isomorphic, but Theorem~\ref{thm:orbiequi} combined with \cite[Prop.\,4.2]{cr1210.6363} implies 
$$
(a,A \otimes A' \otimes A) \cong (a, A' )
\quad \text{in } \Beq . 
$$
Also note that the case $A'=I_a$ says that $a\in \B \subset \Beq$ is the $(A^\star \otimes A)$-orbifold of $(a,A)$, which in turn is the $A$-orbifold of~$a$.

\subsubsection*{Acknowledgements} 

We thank 
Matthias Gaberdiel, 
Andreas Recknagel, 
and 
Ingo Runkel
for helpful discussions. 

\appendix

\section{Discrete torsion from the perspective of boundary CFT}\label{app:CFT}

Landau-Ginzburg models with a quasi-homogenous superpotential are conjectured to flow to superconformal field theories with $\mathcal N=(2,2)$ supersymmetry in the infrared. Therefore, the boundary conditions and defects in discrete torsion orbifolds considered in this paper should have an interpretation within conformal field theory as well. In this appendix, we explain this in the example of a tensor product of two minimal models, repeating the analysis of Section~\ref{subsec:inducedmodules} in CFT language. Here, we can build on \cite{bg0503207}.

Recall that $\mathcal N=2$ minimal models occur at central charge
$$
c=\frac{3k}{k+2} \, .
$$
The bosonic subalgebra of the $\mathcal N=2$ algebra is then realised as the coset
$$
(\mathcal N=2)_{\mathrm{bos}} = \frac{\mathfrak{su}(2)_k \oplus \mathfrak{u}(1)_4}{\mathfrak{u}(1)_{2k+4}} \, .
$$
Accordingly, representations of this coset are labelled by integers $(l,m,s)$, where $l\in \{0, \dots, k\}$ labels the $\mathfrak{su}(2)_k$ representation, $m\in \Z_{2k+4}$ and $s \in \Z_4$ the $\mathfrak{u}(1)$ representations. Furthermore, these integers are subject to the constraint $l+m+s\in 2\Z$ and field identification relates the representations
$$
(l,m,s) \sim (k-l,m+k+2, s+2) \, .
$$
The label $s$ distinguishes NS and R sectors, $s$ even corresponds to the NS sector, $s$ odd to the R sector. 

The Landau-Ginzburg model with superpotential $W=x_1^d+x_2^d$ corresponds to a tensor product of two $\mathcal N=2$ minimal models at level $k=d-2$. More precisely, the Hilbert space of the corresponding conformal field theory is
\begin{align}
\H = & \bigoplus_{[l_1,m_1,s_1],[l_2,m_2,s_2]}
\Big( \left(\H_{[l_1,m_1,s_1]} \otimes \H_{[l_2,m_2,s_2]} \right) 
\otimes \left( \bar\H_{[l_1,m_1,s_1]} \otimes \bar\H_{[l_2,m_2,s_2]}
\right) 
\nonumber  \\
&  \quad
\oplus \,
\left(\H_{[l_1,m_1,s_1]} \otimes \H_{[l_2,m_2,s_2]} \right) 
\otimes \left( \bar\H_{[l_1,m_1,s_1+2]} \otimes 
\bar\H_{[l_2,m_2,s_2+2]} \right) \Big) \label{GSO1}
\end{align}
where the sums over $s_1, s_2$ are subject to spin alignment, i.\,e.~$s_1=s_2$ mod~$2$. Note that this is the Hilbert space of a GSO-projected theory. The other possible GSO-projection corresponds to the superpotential $W=x_1^d+x_2^d+z^2$. The permutation matrix factorisations considered in Example~\ref{ex:inducedcharges} correspond to permutation boundary states. Here, the boundary conditions imposed on the $\mathcal N=2$ generators are
\begin{align}
\left(L^{(1)}_n - \bar{L}^{(2)}_{-n} \right) 
|\!| B \rangle\!\rangle =
\left(L^{(2)}_n - \bar{L}^{(1)}_{-n} \right) 
|\!| B \rangle\!\rangle 
& = 0 \, , 
\nonumber \\
\left(J^{(1)}_n + \bar{J}^{(2)}_{-n} \right) |\!| B \rangle\!\rangle 
= \left(J^{(2)}_n + \bar{J}^{(1)}_{-n} \right) |\!| B \rangle\!\rangle 
& =  0 \, ,
\nonumber \\
\left(G^{\pm (1)}_r + \I \eta_1 \bar{G}^{\pm (2)}_{-r} \right) 
|\!| B \rangle\!\rangle = 
\left(G^{\pm (2)}_r + \I \eta_2 \bar{G}^{\pm (1)}_{-r} \right) 
|\!| B \rangle\!\rangle
& = 0 \, . \label{gluingp}
\end{align}
The relevant boundary states $|\!|  [  L,M,S_1,S_2 ] \rangle\!\rangle$ were constructed in \cite{bg0503207} and read
\be\label{eq:permstates}
\frac{1}{2\, \sqrt{2}} \sum_{l,m,s_1,s_2} 
\frac{S_{Ll}}{S_{0l}} \, \E^{\I\pi M m / (k+2)}\, 
\E^{-\I\pi (S_1 s_1 - S_2 s_2)/2}\,
|[l,m,s_1]\otimes [l,-m,-s_2]\rangle\!\rangle^\sigma  
\ee
where the sum runs over all $l,m,s_1,s_2$ for which 
$$
l+m+s_1 \quad \hbox{and} \quad s_1-s_2 \quad \hbox{are even.}
$$
The Ishibashi states $|[l,m,s_1]\otimes [l,-m,-s_2]\rangle\!\rangle^\sigma$ are in the sectors
$$
\Bigl(\H_{[l,m,s_1]}\otimes \H_{[l,-m,-s_2]}\Bigr) \otimes
\Bigl(\bar\H_{[l,m,s_2]}\otimes \bar\H_{[l,-m,-s_1]}\Bigr) ,
$$
and the labels $L,M, S_1, S_2$ are constrained by the requirement $L+M+S_1-S_2$ even. Shifting one of the $S$-labels by~2 maps a brane to its antibrane; the boundary state is invariant under a shift by~2 of both $S$-labels.

We are interested in permutation boundary states in a $(\Z_{k+2}\times \Z_{k+2})$-orbifold, possibly with discrete torsion. The generators $g_i$ of each $\Z_{k+2}$ act on the Hilbert space by phase multiplication
$$
\left. g_i \right|_{\H_{[l_1,m_1,s_1]}\otimes \H_{[l_2,m_2,s_2]} 
\otimes 
\bar\H_{[l_1,m_1,s_1']} \otimes \bar\H_{[l_2,m_2,s_2']}}
= \exp\left(2\pi \I \frac{m_i}{k+2} \right) \,.
$$
The boundary state is invariant under the diagonal subgroup $\Z_d \subset \Z_d \times \Z_d$. Thus, to discuss boundary conditions in the diagonal $\Z_d$-orbifold, the boundary state has to be formulated on circles twisted by $g^n=(g_1g_2)^n$. The relevant Ishibashi states are built on states from the $g^n$-twisted sector
\begin{align*}
\H_n = & \bigoplus_{[l_1,m_1,s_1],[l_2,m_2,s_2]}
\Bigl( \left(\H_{[l_1,m_1,s_1]} \otimes \H_{[l_2,m_2,s_2]} \right) 
\otimes \left( \bar\H_{[l_1,m_1-2n,s_1]} \otimes \bar\H_{[l_2,m_2-2n,s_2]}
\right)  \\
&  \quad
\oplus \,\, 
\left(\H_{[l_1,m_1,s_1]} \otimes \H_{[l_2,m_2,s_2]} \right) 
\otimes \left( \bar\H_{[l_1,m_1-2n,s_1+2]} \otimes 
\bar\H_{[l_2,m_2-2n,s_2+2]} \right) \Bigr) . 
\end{align*}
B-type permutation gluing conditions require $m_2=-\bar{m}_1= -m_1+2n$ and $m_1=-\bar{m}_2 = -m_2+2n$. 
The boundary state on a twisted circle then reads:
\begin{align*}
& |\!|  [ L,M,\hat{M}, S_1,S_2 ] \rangle\!\rangle \\
= & \, \frac{\E^{-\frac{\pi \I n}{k+2}(M+\hat{M})}}{2\, \sqrt{2}} \sum_{l,m,s_1,s_2} 
\frac{S_{Ll}}{S_{0l}} \, \E^{\I\pi M m / (k+2)}\, 
\E^{-\I\pi (S_1 s_1 - S_2 s_2)/2}\,
|[l,m,s_1]\otimes [l,-m+2n,-s_2]\rangle\!\rangle^\sigma  
\,,\nonumber 
\end{align*}
Here, $\hat{M}$ arises as an additional label for the boundary states and we require that $M+\hat{M}$ is even; 
$\hat{M}$ corresponds to the label $p$ in Example~\ref{ex:inducedcharges}.
The new label specifies a representation on the Chan-Paton labels, which is 1-dimensional in this case. Compared to~\eqref{eq:permstates}, there is an insertion of $g^n$ in the open string channel of the 1-loop amplitude. We refer to \cite{bg0503207} for explicit formulas. 

To formulate a boundary state in the $(\Z_d \times \Z_d)$-orbifold, we have to specify the action of $g_1$, $g_2$ on the twisted sector states. 
We single out a theory without discrete torsion, where $g_i$ acts as
$$
\left. g_i \right|_{\H_{[l_1,m_1,s_1]}\otimes \H_{[l_2,m_2,s_2]} \otimes \bar\H_{[l_1,m_1-2n,s_1']} \otimes \bar\H_{[l_2,m_2-2n,s_2']}}
= \exp\left(2\pi \I \frac{m_i-n}{k+2} \right) .
$$
In this theory, $g_1$ acts on the boundary states in all sectors as
$$
|\!|  [  L,M,\hat{M}, S_1,S_2 ] \rangle\!\rangle_{g^n} \lmt
|\!|  [  L,M +2,\hat{M}, S_1,S_2] \rangle\!\rangle_{g^n} \, . 
$$
Completely parallel to our discussion in terms of matrix factorisations, we can now formulate a consistent boundary state of the $(\Z_d\times \Z_d)$-orbifold by summing over all $g^n$ twisted sectors and taking the $g_1$-orbit: 
$$
|\!|  [  L,M,\hat{M}, S_1,S_2 ] \rangle\!\rangle_{\Z_d\times \Z_d} =
\frac{1}{{\cal N}}\sum_{n,r} \varepsilon(g_1^r, g^n) g_1^r |\!|  [  L,M,\hat{M}, S_1,S_2 ] \rangle\!\rangle_{g^n} \, ,
$$
where $1/{\cal N}$ is a normalisation factor.
This can be evaluated further by inserting the discrete torsion factor $\varepsilon_c(g_1^r, g^n) = \omega^{rnc}$. Hence, the boundary state becomes
\begin{equation}\label{eq:bssum}
|\!|  [  L,M,\hat{M}, S_1,S_2 ] \rangle\!\rangle_{\Z_d\times \Z_d} =
\frac{1}{{\cal N}}\sum_{n,r}  |\!|  [  L,M+2r,\hat{M}-2rc, S_1,S_2 ] \rangle\!\rangle_{g^n} \, .
\end{equation}

By construction, only Ishibashi states from sectors invariant under the respective orbifold projection contribute. RR-charges can be read off as prefactors of the Ishibashi states. It is evident that these charges agree (up to a normalisation factor) with those of the  brane of the diagonal orbifold, provided that the RR ground state survives the respective projection in the closed string sector. The projections in the open string sector can be determined by computing the one-loop formulas; we refer to \cite{bg0503207}, where the partition functions of the individual terms of the sum \eqref{eq:bssum} are given explicitly. Note that in the special case that $c=0$ (no discrete torsion) the result is independent of the label $M$ and we obtain the permutation boundary states of the mirror theory.

To include a non-trivial $\kappa$ in the CFT discussion, we have to modify the group action in the RR sector, altering all projections. As before, we choose an orbifold action with $\kappa_{(0,1)}=\kappa_{(1,0)}=1$. The action of $g_i$ is then modified by a factor of $(-1)^{s_i}$, where because of spin alignment $(-1)^{s_1}=(-1)^{s_2}$. The diagonal action by $g_1 g_2$ remains unchanged, such that we only have to include a factor $(-1)^{s_1}$ in the final projection formula:
$$
|\!|  [  L,M,\hat{M}, S_1,S_2 ] \rangle\!\rangle_{\Z_d\times \Z_d, \kappa} =
\frac{1}{{\cal N}}\sum_{n,r}  |\!|  [  L,M+2r,\hat{M}-2rc, S_1+2r,S_2 ] \rangle\!\rangle_{g^n} \, .
$$
Note that the shift $S_1 \mapsto S_1+2$ maps branes to antibranes, again in complete agreement with the matrix factorisation analysis.

\section{Quantum symmetry defects for cyclic groups}\label{app:qsdefect}

In this appendix, we show that the defect $\mathcal A_G^{\text{qs}}=A_G \otimes A_G$ is isomorphic as an $A_G$-bimodule to the quantum symmetry defect constructed in \cite[Sect.\,4.3]{br0712.0188} for the case $W=x^d$ and $G=\Z_d$. The latter defect is built from $|G|$ copies of $A_G$, with each copy carrying a different representation label of the left $A_G$-action. An equivalent way of writing this is $\bigoplus_{h\in G} {}_{\gamma_{A_G}^h} \!\!(A_G)$, which follows from the fact that shifting the representation label of any left $A_G$-module $X$ by~1 is equivalent to twisting the left action by the Nakayama automorphism $\gamma_{A_G}$ as
$$
\begin{tikzpicture}[very thick,scale=0.75,color=blue!50!black, baseline]
\draw (-0.4,-0.98) node[right] (X) {{\small$\;\;\,{}_{\gamma_{A_G}} X$}};
\draw (-0.4,0.78) node[right] (Xu) {{\small$\;\;\,{}_{\gamma_{A_G}} X$}};

\draw[color=green!50!black] (-0.75,-1) node[left] (B) {{\small$A_G$}};

\draw (0,-1) -- (0,1); 

\fill[color=green!50!black] (0,0) circle (2.5pt) node (meet) {};

\draw[color=green!50!black] (-0.75,-1) .. controls +(0,0.5) and +(-0.5,-0.25) .. (0,0);
\end{tikzpicture} 
=
\begin{tikzpicture}[very thick,scale=0.75,color=blue!50!black, baseline]

\draw (0,-1) node[right] (X) {{\small$X$}};
\draw (0,1) node[right] (Xu) {{\small$X$}};

\draw[color=green!50!black] (-0.75,-1) node[left] (B) {{\small$A_G$}};

\draw (0,-1) -- (0,1); 

\fill[color=green!50!black] (0,0) circle (2.5pt) node (meet) {};

\draw[color=green!50!black] (-0.75,-1) .. controls +(0,0.5) and +(-0.5,-0.25) .. (0,0);

\fill[color=green!50!black] (-0.62,-0.5) circle (2.5pt) node[left] (meet) {{\small$\gamma_{A_G}$}};

\end{tikzpicture} 
\, , 
$$
where
$$
\gamma_{A_G}=\sum_{g\in G} \det(g) \cdot 1_{{}_g I} \, , 
\quad 
\det(g) = \E^{2 \pi \I g /d}
$$
in the present case.

We now prove that for an arbitrary left $A_G$-module $X$ there is an isomorphism
$$
(A_G \otimes A_G) \otimes_{A_G} X \cong \bigoplus_{h\in G} {}_{\gamma_{A_G}^h} \!\!(A_G) \otimes_{A_G} X \, ,
$$
which is equivalent to the statement $A_G \otimes A_G \cong \bigoplus_{h\in G} {}_{\gamma_{A_G}^h} \!\!(A_G)$ (the former implies the latter by setting $X=A_G$). The twisted left unit actions provide an isomorphism
$$
\sum_{g\in G} {}_g (\lambda_X): \, A_G \otimes X \ \lra \ \bigoplus_{g \in G} {}_g X \, ,
$$
and we further have for every $h \in G$
\be\label{eq:qsfusion}
{}_{\gamma_{A_G}^h} \!\!(A_G) \otimes_{A_G} X \cong {}_{\gamma_{A_G}^h} \!\!(A_G \otimes_{A_G} X) \cong {}_{\gamma_{A_G}^h} \!\!X \, .
\ee
It thus suffices to prove $\bigoplus_{g \in G} {}_g X \cong \bigoplus_{h\in G} {}_{\gamma_{A_G}^h} \!\!X$. 

Letting $\gamma$ denote the representation of $G$ on $X$, we construct the map
\begin{align*}
& \varphi = \sum_{g,h\in G} \varphi_{h,g}: \  \bigoplus_{g \in G} {}_g X \ \lra \ \bigoplus_{h\in G} {}_{\gamma_{A_G}^h} \!\!X \, , \\
& \varphi_{h,g}: {}_g X \lra {}_{\gamma_{A_G}^h} \!\!X \, , \quad x \lmt \det(g)^h \gamma_g(x) \, ,
\end{align*}
which has an inverse $\varphi^{-1}= \frac{1}{|G|}\sum_{g,h\in G} \det(g)^{-h} \gamma_g^{-1}$. It remains to show that $\varphi$ is a left $A_G$-module map. The representation of $G$ on $\bigoplus_{g \in G} {}_g X$ is given by 
$$
\Gamma_{g'} = \sum_{g\in G} \pi_{g,g' g} \, , \quad g' \in G \, ,
$$
where $\pi_{g,g' g}: {}_{g'} ({}_g X) \cong {}_{g' g} X$ is simply a permutation of the summands, while the representation on $\bigoplus_{h\in G} {}_{\gamma_{A_G}^h} \!\!X$ is
$$
\widetilde\Gamma_{h'} = \sum_{h\in G} \det(h')^h \, \gamma_{h'} \, , \quad h' \in G \, .
$$
One then directly verifies that the diagram
$$
\begin{tikzpicture}[
			     baseline=(current bounding box.base), 
			     >=stealth,
			     descr/.style={fill=white,inner sep=3.5pt}, 
			     normal line/.style={->}
			     ] 
\matrix (m) [matrix of math nodes, row sep=4em, column sep=2.5em, text height=1.5ex, text depth=0.9ex] {%
{}_{g'g} X && {}_{\gamma_{A_G}^{h}} \!\!X\\
{}_{g'} ({}_g X) && {}_{g'}({}_{\gamma_{A_G}^{h}} \!\!X) \\
};
\path[font=\small] (m-1-1) edge[->] node[auto] {$ \varphi_{h,g'g} $} (m-1-3);
\path[font=\small] (m-2-1) edge[->] node[auto] {$ {}_{g'} (\varphi_{h,g}) $} (m-2-3) 
				  (m-2-1) edge[->]  node[auto] {$ \Gamma_{g'} $} (m-1-1); 
\path[font=\small] (m-2-3) edge[->] node[auto, swap] {$ \widetilde\Gamma_{g'} $} (m-1-3); 
\end{tikzpicture}
$$
commutes for every $g,g',h\in G$, establishing that $\varphi$ is a module map.

Let us finally note that writing the quantum symmetry defect of \cite{br0712.0188} in terms of the Nakayama twists allows us to straightforwardly determine its action on bulk fields and boundaries. Indeed, the action of ${}_{\gamma_{A_G}^h} \!\!(A_G)$ on RR bulk fields \cite[Sect.\,3.4]{BCP1} is given by
$$
\Hom(I,{}_g I) \ni \phi_g \lmt
\begin{tikzpicture}[very thick,scale=0.75,color=blue!50!black, baseline]
\draw (0,0) circle (1.25);
\fill (-52:1.3) circle (0pt) node[right] {{\small${}_{\gamma_{A_G}^h} \!\!(A_G)$}};
\draw[<-, very thick] (0.100,-1.25) -- (-0.101,-1.25) node[above] {}; 
\draw[<-, very thick] (-0.100,1.25) -- (0.101,1.25) node[below] {}; 
\fill[color=green!50!black] (135:0) circle (2.5pt) node[left] {{\small$\phi_g$}};
\draw[color=green!50!black] (0,0) .. controls +(0,0.6) and +(-0.4,-0.4) .. (45:1.25);
\fill[color=green!50!black] (45:1.25) circle (2.5pt) node[right] {};
\fill[color=green!50!black] (25:1.25) circle (2.5pt) node[right] {};
\draw[color=green!50!black] (25:1.25) .. controls +(0.3,0.4) and +(0,-0.5) .. (40:2.3);
\fill[color=green!50!black] (40:2.3) circle (0pt) node[above] {{\small$A_G$}};
\end{tikzpicture} 
\!\!\!\!\!
=
\det(g)^h
\begin{tikzpicture}[very thick,scale=0.75,color=blue!50!black, baseline]
\draw[color=green!50!black] (0,0) circle (1.25);
\fill[color=green!50!black] (-45:1.3) circle (0pt) node[right] {{\small$A_G$}};
\draw[<-, very thick, color=green!50!black] (0.100,-1.25) -- (-0.101,-1.25) node[above] {}; 
\draw[<-, very thick, color=green!50!black] (-0.100,1.25) -- (0.101,1.25) node[below] {}; 
\fill[color=green!50!black] (135:0) circle (2.5pt) node[left] {{\small$\phi_g$}};
\draw[color=green!50!black] (0,0) .. controls +(0,0.6) and +(-0.4,-0.4) .. (45:1.25);
\fill[color=green!50!black] (45:1.25) circle (2.5pt) node[right] {};
\fill[color=green!50!black] (25:1.25) circle (2.5pt) node[right] {};
\draw[color=green!50!black] (25:1.25) .. controls +(0.3,0.4) and +(0,-0.5) .. (40:2.3);
\fill[color=green!50!black] (40:2.3) circle (0pt) node[above] {{\small$A_G$}};
\end{tikzpicture} 
\!\!\!\!\!
= \det(g)^h \pi^{\text{RR}}_{A_G}(\phi_g)
$$
which reproduces the action of the quantum symmetry group as described e.\,g.~in \cite{ginsparg}. For the fusion of ${}_{\gamma_{A_G}^h} \!\!(A_G)$ with any left $A_G$-module $X$ we directly use the result~\eqref{eq:qsfusion}. From the discussion at the beginning of this section it then follows that ${}_{\gamma_{A_G}^h} \!\!(A_G)$ acts on $X$ by shifting its representation label by $h$, just as was found in \cite{br0712.0188}.

\end{document}